\documentclass[a4paper,journal,onecolumn]{IEEEtran}
\IEEEoverridecommandlockouts \makeatletter
\def\ps@headings{%
\def\@oddhead{\mbox{}\scriptsize\rightmark \hfil \thepage}%
\def\@evenhead{\scriptsize\thepage \hfil \leftmark\mbox{}}%
\def\@oddfoot{}%
\def\@evenfoot{}}

\makeatother
\usepackage{algorithm}
\usepackage{algpseudocode}

\usepackage{longtable}
\usepackage{nicematrix}
\usepackage[english]{babel}
\usepackage{fontenc}
\usepackage{placeins}
\usepackage{tikz}
\usetikzlibrary{patterns}
\usepackage{subfigure}
\usepackage{amssymb}
\usepackage{float}
\usepackage{booktabs}
\usepackage{amsmath,amsthm}
\usepackage{cleveref}
\usepackage{colortbl}
\usepackage{makecell}
\usepackage{hhline}
\usepackage{color}
\usepackage{caption}
\usepackage{diagbox}
\usetikzlibrary{matrix}
\usepackage{threeparttable}
\usepackage{array}
\usepackage{arydshln}
\usepackage{pgfplots}

\usepackage[numbers,sort&compress]{natbib}

\usepackage{multirow}
\usepackage{bm}
\usepackage{soul}
\usepackage[normalem]{ulem}
\usepackage{ifthen}

\usepackage{afterpage}
\newcommand{\mr}{\mathrm}
\newcommand{\mb}{\mathbf}
\newcommand{\mc}{\mathcal}

\newcommand{\h}{\mb{h}}

\newcommand{\showrev}{0} 
\newcommand{\rv}[1]{%
	\ifthenelse{\equal{\showrev}{1}}
	{\hl{#1}}%
	{#1}%
}
\newcommand{\rem}[1]{%
	\ifthenelse{\equal{\showrev}{1}}%
	{\textcolor{red}{\sout{#1}}}
	{}%
}
\newcommand{\mrv}[1]{%
	\ifthenelse{\equal{\showrev}{1}}%
	{\colorbox{yellow}{$\displaystyle #1$}}
	{#1}%
}

\newcommand{\tridiagbox}[3]{%
 \begin{tikzpicture}[baseline=(current bounding box.center)]
  \draw (0,0.7) -- (2,0);        
  \draw (1.1,1) -- (2,0);        
  \node at (0.4,0.2) {#1};     
  \node at (0.8,0.7) {#2};    
  \node at (1.7,0.8) {#3};    
 \end{tikzpicture}%
}

\newcommand{\br}{\mb{r}}

\newcommand{\SetR}{\mathbb{R}}
\newcommand{\SetZ}{\mathbb{Z}}

\makeatletter

\newcommand{\Rmnum}[1]{\expandafter\@slowromancap\romannumeral #1@}
\makeatother

\newtheorem{theorem}{Theorem}[subsection]
\newtheorem{lemma}{Lemma}

\newtheorem{claim}{Claim}

\title{ Entropy Functions on Two-Dimensional
  Faces of Polymatroidal Region of Degree Four: Part \Rmnum{1}:
  Problem Formulation and \rv{More} }
\author{Shaocheng~Liu, and Qi~Chen,~\IEEEmembership{Member,~IEEE,} \thanks{
This work was supported in part by NSFC under Grant 62471369 and Grant
62461160306. An earlier version of this paper was presented in part at the IEEE ISIT2023 \cite{LC2023isit}. (Corresponding author: Qi Chen.)
}\\Email: lsc@stu.xidian.edu.cn,
  qichen@xidian.edu.cn
\\School of Telecommunications Engineering, Xidian University}

\makeatletter

\@addtoreset{equation}{subsection}
\makeatother
 
\begin{document}

\maketitle

\begin{abstract}
Characterization of entropy functions is of fundamental importance in
information theory. By imposing constraints on their
Shannon outer bound, i.e., the polymatroidal region, one obtains the faces
of the region and entropy functions on them with special
structures. In this series of two papers, we characterize entropy functions on the $2$-dimensional
faces of the polymatroidal region of degree $4$. 
In Part I, we formulate the problem, enumerate all $59$ types 
of $2$-dimensional faces of the region by an
algorithm, and fully characterize entropy
functions on $49$ types of
them. 
The entropy functions on the remaining $10$ types of
faces will be characterized in Part II, among which $8$ types are
fully characterized, and $2$ types are partially characterized.

Index terms - entropy function, entropy region, polymatroid, poymatroidal region
\end{abstract}

\section{Introduction}
\label{int}

Let $N_n=\{1,2,\ldots,n\}$ and $\mb{X}\triangleq(X_i,i\in N_n)$ be a random vector indexed by
$N_n$. The set function $\h: 2^{N_n}\to \SetR$ defined by
\begin{equation*}
  \h(A)=H(X_A), \quad A\subseteq N_n
\end{equation*}
is called the \emph{entropy function} of $\mb{X}$, while $\mb{X}$ is called a \emph{characterizing random vector} of $\h$.
The Euclidean
space $\mc{H}_n\triangleq\SetR^{2^{N_n}}$ where entropy functions live is called the
\emph{entropy space} of degree $n$.
The set of all entropy functions, denoted by $\Gamma^*_n$, is called
the \emph{entropy region} of degree $n$. 
The characterization of entropy functions, i.e., determining whether an
$\h\in \mc{H}_n$ is in $\Gamma^*_n$, is of
fundamental importance in information theory.

Entropy functions are (the rank functions of) polymatroids, i.e., they satisfy polymatroidal
axioms, that is, for all $A,B\subseteq N_n$,
\begin{align}
  \h(A) &\ge 0,   \label{a1}\tag{I.1}\\
  \h(A)&\le \h(B),\quad \text{ if }  A\subseteq B,\tag{I.2}\\
  \h(A)+\h(B)&\ge \h(A\cap B)+\h(A\cup B).   \label{a3}\tag{I.3}
\end{align}
The region in $\mc{H}_n$ bounded by such inequalities, denoted by $\Gamma_n$, is called the
\emph{polymatroidal region} of degree $n$. Thus, $\Gamma_n$ is an
outer bound on $\Gamma^*_n$. For more about entropy functions, we
referred the readers to \cite[Chapter 13-15]{yeung2008information}.

Traditionally, entropy functions are characterized by information
inequalities. Those inequalities derived by polymatroidal axioms are called Shannon-type,
as they correspond to the non-negativity of Shannon information
measures.  Since 1998, a series of non-Shannon-type information
inequalities, among which Zhang-Yeung inequality is the first one \cite{ZY98},
were discovered \cite{zhang2002new, makarychev2002new, dougherty2011nonshannon}. Thus $\overline{\Gamma^*_n}$, the closure of $\Gamma^*_n$, is strictly included in $\Gamma_n$ when $n\ge 4$.    
Each information inequality determines an outer bound on $\Gamma^*_n$,
as those set functions in $\mc{H}_n$ dissatisfy it must be located outside
$\Gamma^*_n$. In this series of two papers, we develop a system of
entropy function characterization from the perspective of faces of
$\Gamma_n$, which covers the traditional inequality characterization. 

By definition, $\Gamma_n$ is a polyhedral cone in $\mc{H}_n$.
Thus, each Shannon-type information inequality determines a face $F$ of
$\Gamma_n$. It is natural to
characterize entropy functions on the specific $F$ of
$\Gamma_n$ (See Subsection \ref{fpc} for details on the faces of a
polyhedral cone). Let $F^*\triangleq F\cap \Gamma^*_n$ be the set of all entropy
functions in $F$. In what follows, when we refer to determining the entropy functions on $F$
, or the region $F^*$, we will simply call this characterizing $F$.
A non-Shannon-type information inequality 
can be considered as an outer bound on $F^*$ when
$F=\Gamma_n$ itself, the improper face of $\Gamma_n$. 
A constrained non-Shannon-type information inequality is an outer bound on $F^*$
when $F$ is the face determined by the constraints that are equalities
obtained by setting some Shannon-type inequalities as equalities. When $F$
is an extreme ray, i.e., a $1$-dimensional face of $\Gamma_n$, if it
contains a matroid, entropy functions on $F$ are called matroidal
entropy functions, and they can be fully characterized by the
probabilistically characteristic set of the matroid \cite{CCB21,CCB22,CCB24}. Mat{\'u}{\v{s}}
fully characterized the first non-trivial 2-dimensional face of $\Gamma_n$
in 2006 \cite{matus2005piecewise}. It is a $2$-dimensional face of $\Gamma_3$. In
2012, Chen and Yeung  characterized another $2$-dimensional face of $\Gamma_3$ \cite{chen2012characterizing}. They are the
only two types of non-trivial $2$-dimensional faces of $\Gamma_3$ that need to be
characterized. To the best of our knowledge, so far, there is no fully characterized
non-trivial $3$-dimensional faces. However,
partial characterizations of $3$-dimensional faces of $\Gamma_3$ can
be found in \cite{HCG12,11003168}.

Many information-theoretic problems can be considered  as applications
of entropy function characterizations on faces of $\Gamma_n$. In a
series of three
papers \cite{matuvs1995a, matuvs1995b, matu1999conditional},
Mat{\'u}{\v{s}} and Studne{\'y} solved the probabilistic conditional independence problem
for four random variables. Note each class of conditional independence
constraints, which is called a semimatroid in their papers, 
determines a face $F$ of $\Gamma_n$.  
The solution to this problem, i.e., the
probabilistic-representability of a semimatroid 
determines whether the relative interior of the corresponding face $F$ intersets with
$\Gamma^*_n$.
This problem thus can be considered as partial characterizations of
the faces of $\Gamma_n$. 
In \cite{yan2012implicit},
Yan, Yeung, and Zhang proved a formula involving $\Gamma_n^*$ for the capacity of multi-source
multi-sink network coding. Those constraints in the formula induced
by network topology and source independency form a face of $\Gamma_n$,
which shows that this holy-grail network coding problem corresponds to
the entropy function characterizations on this face. For the secret-sharing problem, see \cite{beimel2011secret} for example, the perfect
secrecy and decoding correctness conditions of an access structure
determine a face $F$ on
$\Gamma_n$, and the information ratio of the secret-sharing problem
can be considered as an optimization problem whose feasible region
is $F\cap \Gamma^*_n$. Other problems, such as distributed data storage
\cite{tian2014characterizing}, coded caching \cite{tian2018symmetry}, 
Markov random fields \cite{yeung2018information}, and relational database \cite{Dan2023} are also related to the
entropy regions on the faces of $\Gamma_n$. Recently, there are a series of
works on how to compute the faces of $\Gamma_n$ more efficiently by machine \cite{PSITIP,YL21,GYG23,GYG25}.

Though the information theory problems above usually involve more
  than four random variables, and the corresponding faces are of
  dimension higher than $2$, following the characterizations of
  extreme rays. i.e., $1$-dimensional faces containing matroid
\cite{CCB21,CCB22,CCB24}, and
$2$-dimensional faces of $\Gamma_3$
\cite{matus2005piecewise,chen2012characterizing},
in this series of two papers,  we characterize entropy functions on the $2$-dimensional
faces of $\Gamma_4$, which may serve as stepping stones to the
  general cases of this problem. 

In Part I, \rv{we first characterize entropy functions on extreme rays
  containing integer polymatroids $\hat{U}_{2,5}$ and
  $\hat{U}_{3,5}$ (see Subsection} \mrv{\ref{fpc}} \rv{for the definition),
which cannot be find in any existing literatures.} We then enumerate all $59$ types of $2$-dimensional faces
of $\Gamma_4$ by an algorithm and completely characterize $49$ types of
them. Among them, $13$ types are all-entropic faces, while $7$
are non-entropic ones\rv{\footnote{\rv{By} \mrv{\cite{matuvs1995a,
        matuvs1995b, matu1999conditional}}, \rv{all $2$-dimensional faces
    containing  V\'amos ray $V_8$ are non-entropic, and they are also
  all of the $2$-dimensional faces crosscut by Zhang-Yeung, or any
  other non-Shannon-type inequalities.}}}. The characterization of $17$ types can be obtained by the $2$ types of non-trivial faces of $\Gamma_3$ characterized in
\cite{matus2005piecewise} and \cite{chen2012characterizing}. 
Of the $22$ types of non-trivial faces, $12$ will be
characterized in Part I. 
The primary technique we adopt to characterize
these non-trivial faces is 
the graph-coloring in \cite{matus2005piecewise}. For a $3$-dimensional random vector $(X_i,i\in N_3)$, if $X_1$ and $X_2$ are
independent, and $X_3$ is a function of $(X_1,X_2)$, $X_1$ and $X_2$
can be considered distributed on the two sets of vertices $\mc{X}_1$
and $\mc{X}_2$ of a
bipartite graph with edges those $(x_1,x_2)$ such that $p(x_1,x_2)>0$,
and these edges are colored by $x_3\in \mc{X}_3$. In our work, when
the number of the random variables becomes four, we color bipartite graphs
by two color systems or color a tripartite graph by one color system.

The remaining $10$ cases of faces will be characterized in Part II \cite{LCC2024part2} of
the two-paper series. Among them,  $8$ faces are fully  and $2$ are
partially characterized. This part has a more combinatorial
flavour. All integer polymatroid with rank exceeding $1$ contained by an extreme ray will be
associated a combinatorial structure. Then the entropy functions on
$2$-dimensional faces containing $2$ such extreme rays breed new
combinatorial-design structures, generalizing corresponding combinatorial
structures in both extreme rays.
These structures can also be
considered as generalizations of $\rm VOA$s introduced in \cite{CCB21,CCB22}.


The rest of this paper, i.e., Part I, is organized as follows. Section \ref{pre}
gives the preliminaries on matroids, polyhedral cones, and the extreme
rays of $\Gamma_4$. 
All the $59$
types of faces are generated by Algorithm I and listed in Table \ref{table:1}  in Section \ref{enu}. In
Section \ref{af}, we fully characterize $49$ types among all the
$2$-dimensional faces of $\Gamma_4$, and the results are listed in
Table \ref{long}.

\section{Preliminaries}
\label{pre}
\subsection{Matroids}
For a polymatroid $\h\in \Gamma_n$, if $\h(A)\in \SetZ$ for all
$A\subseteq N_n$, $\h$ is called \emph{integer}. 
An ordered pair $M=(N_n,\br)$ is called a \emph{matroid} with \emph{rank function}
$\br$ if $\br$ is an integer polymatroid with $\br(A)\le |A|$ for all
$A\subseteq N_n$. Like polymatroids, in this paper, we do not distinguish a matroid and
its rank function unless otherwise specified. 

A \emph{uniform matroid} $U_{k,n}$ is a matroid with rank function
$\br(A)=\min\{k,|A|\}$ for all $A\subseteq N_n$.

For a matroid $M$ and
$e\in N_n$, if $\br(e)=0$, $e$ is called a \emph{loop} of $M$. For $e, e'\in
N_n$, if $\br(\{e,e'\})=1$, then $e$ and $e'$ are called
\emph{parallel}.

For more
about matroid theory, readers
can refer to \cite{oxley2006matroid}.

\subsection{Faces of a polyhedral cone}
\label{fpc}
Let $C\subseteq \SetR^d$ be a full-dimensional polyhedral cone. For a
hyperplane $P$ containing the origin $O$ in $\SetR^d$, if $C\subseteq P^+$, where
$P^+$ is one of the two halfspaces corresponding to $P$,
$F\triangleq C\cap P$ is called a (proper) face of $C$, while $C$
itself is its improper face. When $\dim F=d-1$, $F$ is called a
\emph{facet} of $C$, and when $\dim F=1$, $F$ is an extreme ray of
$C$. Either the set of all facets or the set of all extreme rays of $C$
uniquely determines the cone, and they are called H-representation and
V-representation of the cone, respectively. Each face $F$ of the
cone  can be written as the intersection of the facets of the cone
that contains $F$, 
or the convex hull of the extreme rays contained in $F$. We also call
them the H-representation and V-representation of $F$,
respectively. More details on this topic are referred to \cite{ziegler2012lectures}.

As discussed in Section \ref{int}, $\Gamma_n$ is a polyhedral cone
in $\mc{H}_n$ determined by polymatroidal axioms in \eqref{a1}-\eqref{a3}.
They are equivalent to the following 
elemental inequalities
\begin{align}
  \label{eq:1}
 &\h(N_n)\ge \h({N_n\setminus \{i\}})\quad\quad  i\in N_n;\\
  &\h(K)+\h(K\cup ij)\le \h(K\cup i)+\h(K\cup j) ,  \label{eq:1a}\\
                         & \quad\quad i,j\in N_n, K\subseteq N_n\setminus\{i,j\}\nonumber
\end{align}
each of which determines a facet of $\Gamma_n$ \cite[Chapter 14]{yeung2008information}.  
When $n=4$, it can be checked that there are $28$ elemental
inequalities, or $28$ facets of $\Gamma_4$. 

It can be seen in \cite{matuvs1994extreme} that 
there are $41$ extreme
rays of $\Gamma_4$. Note that each extreme ray $E_{\br}$ of $\Gamma_4$ can be written in the form
\begin{equation}
\label{eray}
  E_{\br}=\{a\br: a\ge 0\}
\end{equation}
where $\br$ is the minimal integer polymatroid in the ray\rv{, that is,
   an integer polymatroid such that $\br/t$ is not integer for any
   integer $t>1$.
}.

Therefore, in
this paper, we refer to the minimal integer polymatroid $\br$ as the
extreme ray $E_{\br}$ containing it unless otherwise specified. Note
that some of these integer polymatroids are matroids.
The
$41$ extreme rays can be classified into the following $11$
types.

\begin{itemize}
\item $U^{i}_{1,1}$, $i\in N_4$;
\item $U^\alpha_{1,2}$, $\alpha\subseteq N_4$, $|\alpha|=2$;
\item $U^\alpha_{1,3}$, $\alpha\subseteq N_4$, $|\alpha|=3$;
\item $U^\alpha_{2,3}$, $\alpha\subseteq N_4$, $|\alpha|=3$;
\end{itemize}
for $U^\alpha_{k,m}$ with $\alpha\subseteq N_4$ and $|\alpha|=m$, we mean a
matroid on $N_4$ whose restriction on $\alpha$ is a $U_{k,m}$ and
$e\in N_4\setminus \alpha$ being loops; 
\begin{itemize}
\item $U_{1,4}$;
\item $\mc{W}^\alpha_2$, $\alpha\subseteq N_4$, $|\alpha|=2$;
\end{itemize}
for $\mc{W}^\alpha_2$ with
$\alpha\subseteq N_4$, $|\alpha|=2$, it is called a wheel matroid with
order $2$,\footnote{We adopt the notation and terminology in
  \cite[Section 8.4]{oxley2006matroid}} and it is a matroid with two parallel elements in $\alpha$, and
each element in $\alpha$ and the other two elements in $N_4$ form a
$U_{2,3}$;
\begin{itemize}
\item $U_{2,4}$;
\item $U_{3,4}$;
\item $\hat{U}^i_{2,5}$, $i\in N_4$;
   \end{itemize}
for $\hat{U}^i_{2,5}$ with $i\in N_4$, it is a polymatroid whose free
expansion is a $U_{2,5}$,\footnote{See \cite[Theorem
  1.3.6]{nguyen1978semimodular} or \cite[Theorem 4]{matuvs2007}  for
  the definition of free expansion.} and \rv{its}  rank function is defined by 
\begin{equation*}
  \br(A)=
  \begin{cases}
    \min \{2, |A|\},\quad &A\neq \{i\},\\
    2,\quad & A=\{i\},
  \end{cases}
\end{equation*}
for any $A\subseteq N_4$, 
\begin{itemize}
    \item $\hat{U}^i_{3,5}$;
\end{itemize}
for $\hat{U}^i_{3,5}$ with $i\in N_4$, it is a polymatroid whose free
expansion is a $U_{3,5}$, and \rv{its} rank function is defined by 
\begin{equation*}
  \br(A)=
  \begin{cases}
    \min \{3, |A|+1\},\quad & i\in A,\\
    |A|,\quad & i\not\in A,
  \end{cases}
\end{equation*}
for any $A\subseteq N_4$; 
\begin{itemize}
  \item $V^\alpha_8$, $\alpha\subseteq N_4$, $|\alpha|=2$;
  \end{itemize}
for $V^\alpha_8$ with $\alpha\subseteq N_4$ and $|\alpha|=2$, it is a polymatroid whose free
expansion is the V\'amos matroid, and \rv{its} rank function is defined by 
\begin{equation*}
  \br(A)=
  \begin{cases}
   3,\quad & |A|=2 \text{ and } A\neq \alpha, \\
    \min \{4, 2|A|\},\quad & \text{ o.w.}
  \end{cases}
\end{equation*}


It can be seen that for an extreme ray in the form $E^\alpha$ with
$\alpha\subseteq N_4$, it is in a type with $\binom{4}{|\alpha|}$
extreme rays, and each extreme ray in the type can be obtained from
each other by permuting the indices in $N_4$. 
For a specific extreme ray, say $U^{\{1,2\}}_{1,2}$, for
simplicity, we will drop the bracket and comma of the set in the
superscript, and write it as $U^{12}_{1,2}$. 
\rv{To facilitate the readers,
 the $11$ types of extreme rays are presented in Table}
\mrv{\ref{tablevector}} \rv{by the rank functions of their minimal
  integer polymatroid in the form of $15$-dimensional vectors, i.e.,
  $\br=(\br(A),\emptyset\neq A\subseteq N_4)$.}

In the same manner, we
denote the $8$ extreme rays of $\Gamma_3$ and classify them in the following $4$
types.
\begin{itemize}
\item $U^{i,3}_{1,1}$, $i\in N_3$;
\item $U^{\alpha,3}_{1,2}$, $\alpha\subseteq N_3$, $|\alpha|=2$;
\item $U_{1,3}$;
\item $U_{2,3}$.
\end{itemize}
Here we put second superscript $3$ to $U^{i,3}_{1,1}$ and
$U^{\alpha,3}_{1,1}$ to distinguish them from $U^{i}_{1,1}$ and
$U^{\alpha}_{1,1}$, the extreme rays of $\Gamma_4$, respectively.

%

\begin{table*}[h]
	\centering
	\caption{Extreme rays of $\Gamma_4$ and their  rank functions}
	\renewcommand{\arraystretch}{2}
	{\setlength{\tabcolsep}{4pt} \mrv{
	     \begin{tabular}{@{}c@{}|cccc|cccccc|cccc|c}
			\hline \tridiagbox{$E_{M}$}{$\mb{r}_M(A)$}{$A$} & $1$ &$2$& $3$& $4$
			& $12$& $13$& $14$ &$23$& $24$ &$34$ &$123$ &$124$& $134$& $234$& $1234$  \\
			\hline\hline  $U^{1}_{1,1}$& $1$& $0$ &$0$ &$0$ & $$1& $1$ &$1$ &$0$& $0$& $0$& $1$& $1$& $1$& $1$&$1$\\
			\hline $U^{12}_{1,2}$& $1$& $1$ &$0$ &$0$ & $1$& $1$ &$1$ &$1$& $1$& $0$&$1$& $1$& $1$& $1$&$1$\\
			\hline $U^{123}_{1,3}$& $1$& $1$ &$1$ &$0$ & $1$& $1$ &$1$ &$1$& $1$& $1$& $1$&$ 1$& $1$& $1$&$1$\\
			\hline $U_{1,4}$& $1$& $1$ &$1$ &$1$ &$ 1$& $1$ &$1$ &$1$& $1$& $1$& $1$& $1$& $1$& $1$&$1$\\
			\hline $U^{123}_{2,3}$& $1$& $1$ &$1$ &$0$ & $2$& $2$ &$1 $&$2$& $1$& $1$& $2$& $2$& $2$& $2$&$2$\\
			\hline $\mc{W}^{12}_2$& $1$& $1$ &$1$ &$1$ & $1$& $2$ &$2$ &$2$&$ 2$& $2$& $2$& $2$&$ 2$& $2$&$2$\\
			\hline $U_{2,4}$& $1$& $1$ &$1$ &$1$ & $2$& $2$ &$2$ &$2$& $2$&$ 2$& $2$&$ 2$& $2$& $2$&$2$\\
			\hline $U_{3,4}$& $1$& $1$ &$1$ &$1$ & $2$& $2$ &$2$ &$2$& $2$&$ 2$& $3$&$ 3$& $3$& $3$&$3$\\
			\hline $\hat{U}^1_{2,5}$& $2$& $1$ &$1$ &$1$ & $2$& $2$ &$2$ &$2$& $2$& $2$& $2$& $2$&$ 2$& $2$&$2$\\
			\hline $\hat{U}^1_{3,5}$& $2$& $1$ &$1$ &$1$ &$ 3$& $3$ &$3$ &$2$& $2$&$ 2$& $3$& $3$&$ 3$& $3$&$3$\\
			\hline $V^{12}_8$  & $2$&$ 2$ &$2$ &$2$ & $4$& $3$ &$3$ &$3$& $3$& $3$&$ 4$& $4$& $4$&$ 4$&$4$\\
			\hline
		\end{tabular}}
	}	\label{tablevector}
\end{table*}

\subsection{Entropy functions on the extreme rays of $\Gamma_4$}
\label{ray}

As discussed in Subsection \ref{fpc}, the V-representation of a
face can be written as the convex hull of the extreme rays it
contains. Therefore, to characterize entropy functions on
$2$-dimensional faces of $\Gamma_4$, we need first to characterize the
entropy functions on the extreme rays of $\Gamma_4$. The results of the
characterizations of all extreme rays in this subsection are summarized in Table \ref{table12}.

In the following,
extreme rays are written in the form in \eqref{eray}. For each type of
extreme ray\rv{,} we only consider one of them. Throughout this paper,
 for any random
vector $(X_i:i\in N_4)$, each $X_i$ is assumed distributed on a finite
set $\mc{X}_i, i\in N_4$. For any marginal distribution of  $(X_i: i\in
N_4)$, say $p_{X_2 X_3}(x_2,x_3)$ with $(x_2,x_3)\in\mc{X}_2\times
\mc{X}_3$, we usually drop its subscription and write it as
$p(x_2,x_3)$ if there is no ambiguity. We may further simplify it as
\rv{$p(x_{23})$}.

\begin{theorem}
  \label{allentropic}
  For $E=U^1_{1,1}, U^{12}_{1,2}, U^{123}_{1,3}, U_{1,4}$, i.e.,
  extreme rays containing a matroid with rank $1$, $a\br\in E$ is entropic
  for all $a\ge 0$.
\end{theorem}
\begin{proof}
  Let $X_1$ be any random variable with $H(X_1)=a$. Let $X_i=X_1$ for
  any $i\in\alpha$ and $X_j$ is a constant if $j\in
  N_4\setminus\alpha$, where $\alpha=N_k$ if
  $E=U^\alpha_{1,k},k=1,2,3,4$ and here $U^{N_k}_{1,4}=U_{1,4}$.
It is readily seen that the entropy function of such constructed
$(X_i,i\in N_4)$ is $a\h$.
\end{proof}

The characterization of the following four types of extreme rays follows
immediately from matroidal entropy functions in \cite{CCB21} and 

\begin{theorem}
  \label{integerk}
  For $E=U^{123}_{2,3}, \mc{W}^{14}_2, U_{3,4}$, $a\br\in E$ is entropic if and only if
  $a=\log v$ for integer $v\ge 1$.
\end{theorem}

\begin{theorem}(\cite[Proposition 2]{CCB21})
  \label{intn26}
  For $E=U_{2,4}$, $a\br\in E$ is entropic if and only if
  $a=\log v$ for positive integer $v\neq 2,6$.
\end{theorem}

\begin{table*}[h]
	\centering
	\caption{Entropy functions on the extreme rays of $\Gamma_4$}
	\renewcommand{\arraystretch}{1}
	\begin{tabular}{c|c|c|c}
		\hline  Extreme rays $E$& Theorems
		&  Figures of entropy region $E^*=E\cap \Gamma^*_4$ &
                                                                      Entropy
                                                                      region $E^*$
          \\
		\hline \hline	$U^i_{1,1}$& \ref{allentropic} &
                                              \includegraphics[trim=1.6 0 0 0,scale=0.9]{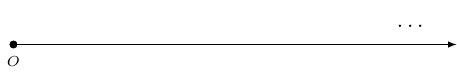} &  $\{a\br:a\ge 0\}$
		\\
		\hline $U^\alpha_{1,2}$& \ref{allentropic}  &\includegraphics[trim=1.6 0 0 0,scale=0.9]{fig_ex1.pdf} &  $\{a\br:a\ge 0\}$      \\
		\hline	$U^\alpha_{1,3}$ & \ref{allentropic} & \includegraphics[trim=1.6 0 0 0,scale=0.9]{fig_ex1.pdf}& $\{a\br:a\ge 0\}$   \\
		\hline	$U_{1,4}$ & \ref{allentropic} &\includegraphics[trim=1.6 0 0 0,scale=0.9]{fig_ex1.pdf} & $\{a\br:a\ge 0\}$  \\
		\hline	$U^\alpha_{2,3}$& \ref{integerk} & \includegraphics[trim=1.6 0 0 0,scale=0.9]{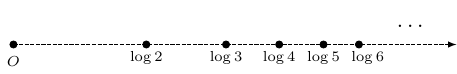} &$\{a\br:a=\log v,\ v\ge 1, v \in\mathbb{Z}\}$  \\
		\hline	$\mc{W}^\alpha_{2}$ & \ref{integerk} &\includegraphics[trim=1.6 0 0 0,scale=0.9]{fig_ex2.pdf} &$\{a\br:a=\log v,\ v\ge 1, v\in \mathbb{Z}\}$\\
		\hline$U_{2,4}$ & \ref{intn26} & \includegraphics[trim=1.6 0 0
                                  0,scale=0.9]{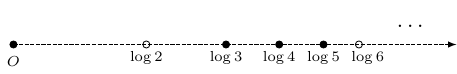} &
                                                              $\{a\br:a=\log v, v\ge 1, v\in \mathbb{Z},v\neq 2,6\}$ \\
		\hline	$U_{3,4}$  &  \ref{integerk}&    \includegraphics[trim=1.6 0 0 0,scale=0.9]{fig_ex2.pdf} & $\{a\br:a=\log v,\  v\ge 1, v\in \mathbb{Z}\}$ \\
		\hline	$\hat{U}^i_{2,5}$ & \ref{rk16}&   \includegraphics[trim=1.6 0 0 0,scale=0.9]{fig_ex2.pdf}& $\{a\br:a=\log v,\  v\ge 1, v\in \mathbb{Z}\}$ \\
		\hline	$\hat{U}^i_{3,5}$ & \ref{rk17} &  \includegraphics[trim=1.6 0 0 0,scale=0.9]{fig_ex2.pdf} & $\{a\br:a=\log v,\  v\ge 1, v\in \mathbb{Z}\}$ \\
		\hline	$V^\alpha_8$& \ref{vray} & \includegraphics[trim=1.6 0 0
                                       0,scale=0.9]{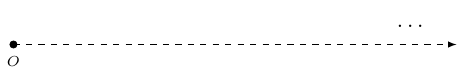}  &
                                                                    $\{a\br: a= 0\}$  \\	
          
		\hline
	\end{tabular}
	\label{table12}
\end{table*} 

It is well known that the rays containing a minimal integer
polymatroid 
with rank-$1$ are all-entropic. It has been proven that the \rv{vectors $a\br$}  in
$U^{\alpha}_{2,3}$ are entropic if and only if $a=\log{v}$ for $k\in
\mathbb{Z}$ \cite[Theorem 15.2]{yeung2008information}. In the ray
$U_{2,4}$, $a$ can take the values $\log{v}$ except for $\log2$ and
$\log6$ \cite[Proposition 6]{CCB22}. In the ray $U_{3,4}$, the region
is the same as in $U^{\alpha}_{2,3}$. These three cases correspond to
the Latin Square, Orthogonal Two Latin Squares, and Latin Cube,
respectively. As for the ray $\hat{U}^i_{2,5}$ and $\hat{U}^i_{3,5}$,
they are implied by the following theorems. For the ray $V^\alpha_8$,
the fact that it disobeys the Zhang-Yeung inequality indicates that it lies outside of $\overline{\Gamma_4^*}$ and, consequently, is non-entropic. 
\begin{theorem}
	\label{rk16}
	For the rank function $\br$ of $\hat{U}^i_{2,5}$, $\h=a\br$  is entropic if and only if $a=\log v$ for integer $v\ge 1$.
	\begin{proof}
		For $\h=a\br$, where $\br$ is
		the rank function of $\hat{U}^1_{2,5}$, if we restrict it on
		$\{2,3,4\}$, we will see that its restriction $\h'=a\br_1$, where $\br_1$ is
		the rank function of $U_{2,3}$, is entropic only if $a=\log v$ for integer $v\ge 1$.
		
		 To prove the ``if'' part, let $X_2$ and $X_3$ be independent and uniformly distributed on $\mathbb{I}_v\triangleq\{0,\dots,v-1\}$, and $X_4=X_2+X_3\mod v$. Let $X_1=v\cdot X_2+X_3$. Then $X_1$ is uniformly distributed on $\mathbb{I}_{v^2}$ and it can be checked that $(X_i,i\in N_4)$ such constructed characterizes $\h=\log{v}\cdot\br$. 
	\end{proof}
\end{theorem}
\begin{lemma}
	\label{lem1}
	Assume that each $X_i,i\in N_3$ is distributed on $\mc{X}_i$ and for  $x_i\in\mathcal{X}_{i}, p(x_{i})>0$.  If $X_1$ and $X_2$ are independent and 
	for any $p(x_1,x_2,x_3)>0$, $p(x_1)=p(x_2)$, then $X_1$ and $X_2$ are uniformly distributed on $\mathcal{X}_{1} $ and $\mathcal{X}_{2}$, respectively, $|\mathcal{X}_{1}|=|\mathcal{X}_{2}|$ and $H(X_1)=H(X_2)$.
\end{lemma}
\begin{proof}
		Let $X_1$ and $X_2$ be independent and $p(x_1,x_2,x_3)>0$. By assumption, we have 
	\begin{align}
		p(x_1) =p(x_2). \label{lem1.1}
	\end{align}
	Since  $X_1$ and $X_2$ are independent, for any $x_1'\in \mathcal{X}_{1}$ with $x_1'\neq x_1$, we obtain
	\begin{equation}
		p(x_1',x_2)=p(x_1')p(x_2)>0 . \label{lem1.2}  \nonumber
	\end{equation}
	As 
	\begin{equation}
		p(x_1',x_2)=\sum_{x_3} p(x_1',x_2,x_3),  \label{lem1.3} \nonumber
	\end{equation}
	there exists $x_3'\in \mc{X}_3$ such that $p(x_1',x_2,x_3')>0$, which implies
	\begin{equation}
		p(x_1')=p(x_2)  \label{lem1.4} 
	\end{equation}
according to the assumption. By \eqref{lem1.1} and \eqref{lem1.4},
	\begin{equation}
		p(x_1)=p(x_1'), \label{lem1.5} 
	\end{equation}
	which implies $p(x_1)$ are equal for all $x_1\in \mc{X}_1$. Then, we conclude that $X_1$ is uniformly distributed on $\mathcal{X}_{1}$. By symmetry, $X_2$ is  uniformly distributed on $\mathcal{X}_{2} $. As $p(x_1)=p(x_2)$ for all $x_1\in\mathcal{X}_{1} $ and $x_2\in\mathcal{X}_{2}$, then $|\mc{X}_1|=|\mc{X}_2|$ and $H(X_1)=H(X_2)$.
\end{proof}
\begin{theorem}
	\label{rk17}
	For the rank function $\br$ of $\hat{U}^i_{3,5}$, $\h=a\br$  is entropic if and only if $a=\log v$ for integer $v\geq 1$.
\end{theorem}
\begin{proof}
	If $\mathbf{ h} \in  F= \hat{U}^1_{3,5}$ is entropic, its characterizing random vector
	$(X_i,i\in N_4)$ satisfies the following information equalities,
	\begin{align}
	H(  X_{N_4})&=H(X_{N_4-i}), \ i\in N_4& \nonumber \\
	H(X_{ij })&=H(X_{i})+H(X_{j}),  i\neq j, i,j\in N_4& \nonumber \\
	H(X_{  ik })+H(X_{jk })&=H(X_{ k})+H(X_{  i,j,k} ),\nonumber\\  &\qquad\qquad\qquad\quad\ \{i,j,k\}=\{2,3,4\} \nonumber \\
	H(X_{i\cup K})+H(X_{  j\cup K})&=H(X_{ K})+H(X_{  ij
		\cup K }),\nonumber\\ &\qquad\qquad\qquad |K|=2, K\nsubseteq\{2,3,4\}.   \nonumber
\end{align}
	Assume that each $X_i$ is distributed on $\mc{X}_i$ and for $x_i\in
	\mathcal{X}_{i}$, $p(x_{i})>0$.  For $ (x_i,i\in N_4) \in
	\mathcal{X}_{N_4}$ with $p(x_{1234})>0$, above information
        equalities
	imply that the probability mass function satisfies
	\begin{align}
		p(x_{1},x_{2},x_{3},x_{4})&=p(x_{1},x_{2},x_{3}) \label{c250}\\
		&=p(x_{1},x_{2},x_{4}) \label{c251} \\ 
		&=p(x_{1},x_{3},x_{4}) \label{c252}\\
		&= p(x_{2},x_{3},x_{4}), \label{c253}\\
		p(x_1,x_2)&=p(x_1)p(x_2), \label{c254}\\
		p(x_1,x_3)&=p(x_1)p(x_3), \label{c255}\\
		p(x_1,x_4)&=p(x_1)p(x_4), \label{c256}\\
		p(x_2,x_3)&=p(x_2)p(x_3), \label{c257}\\
		p(x_2,x_4)&=p(x_2)p(x_4), \label{c258}\\
		p(x_3,x_4)&=p(x_3)p(x_4), \label{c259}\\
		p(x_2,x_3)p(x_2,x_4)&=p(x_2)p(x_2,x_3,x_4), \label{c260}\\
		p(x_2,x_3)p(x_3,x_4)&=p(x_3)p(x_2,x_3,x_4), \label{c261}\\
		p(x_2,x_4)p(x_3,x_4)&=p(x_4)p(x_2,x_3,x_4), \label{c262}\\
		p(x_{1},x_{2},x_{3})p(x_{1},x_{2},x_{4})&=p(x_{1},x_{2})p(x_{1234}),\label{c263}\\
		p(x_{1},x_{2},x_{3})p(x_{1},x_{3},x_{4})&=p(x_{1},x_{3})p(x_{1234}),\label{c264}\\
		p(x_{1},x_{2},x_{4})p(x_{1},x_{3},x_{4})&=p(x_{1},x_{4})p(x_{1234}).\label{c265}
	\end{align}
	By \eqref{c250}-\eqref{c253} and \eqref{c263}-\eqref{c265}, we have 
	\begin{align}
		p(x_{1234})&=p(x_1,x_2,x_3)\label{c266} \\
		&=p(x_1,x_2)=p(x_1,x_3)=p(x_1,x_4).\label{c267}
	\end{align}	 
	In light of \eqref{c254}-\eqref{c256}, relapcing $p(x_1,x_2),p(x_1,x_3)$ and $p(x_1,x_4)$ by $p(x_1)p(x_2),p(x_1)p(x_3)$ and $p(x_1)p(x_4)$ in \eqref{c267}, we obtain 
	\begin{equation}
		p(x_2)=p(x_3)=p(x_4). \label{c268}
	\end{equation}
	By Lemma \ref{lem1}, $X_2,X_3$ and $X_4$ are uniformly distributed on  $\mc{X}_2,\mc{X}_3$ and $\mc{X}_4$ and $H(X_2)=H(X_3)=H(X_4)=\log v$, where $v=|\mc{X}_i|,i\in \{2,3,4\}$.
	As $\h\in F$ and $(X_i, i\in N_4)$ is its characterizing random vector, we have $H(X_2)=H(X_3)=H(X_4)=a$. Thus $a$ can only take the value $\log v$.
	
	Now we prove that all of these polymatroids are entropic. For $a=\log v$, let $X_2,X_3$ and $X_4$ be  mutually independent and uniformly distributed on $\mathbb{I}_v$. Let
	\begin{align}
			X_1=v&[(X_3-X_2)\mod v]\nonumber\\+&[(X_4-X_2)\mod v]. \label{c2000}
	\end{align}
It remains to  check the entropy function of such constructed $(X_i,i\in N_4)$ is $a\br$. Given that $X_2$, $X_3$ and $X_4$ are  mutually independent and  uniformly distributed, we have 
\begin{align}
	H(X_2)&=H(X_3)=H(X_4)=\log v \label{c2010}\\
H(X_{234})&=H(X_2)+H(X_3)+H(X_4)\\&=3\log v. \label{c2011}
\end{align}
Then by \eqref{c2000}, we obtain $H(X_{1234})=H(X_{234})$. Dividing
both sides of \eqref{c2000} by $v$ and then rounding them down, we obtain
\begin{align}
\lfloor\dfrac{X_1}{v}\rfloor= (X_3-X_2)\mod v\label{c2001}
\end{align}
Thus \begin{align}
	X_3 \mod v= (X_2+\lfloor\dfrac{X_1}{v}\rfloor) \mod v.\label{c2002}
\end{align}
Note that $X_3$ is distributed on $\mathbb{I}_v$, we have
\begin{align}
	X_3 = (X_2+\lfloor\dfrac{X_1}{v}\rfloor) \mod v,\label{c2003}
\end{align}
which implies that $X_3$ is a function of $(X_1,X_2)$. By the same argument, 
\begin{align}
	X_2 = (X_3-\lfloor\dfrac{X_1}{v}\rfloor) \mod v.\label{c2004}
\end{align}
  By  \eqref{c2001}, replacing  $(X_3-X_2)\mod v$ by $\lfloor\dfrac{X_1}{v}\rfloor$  in \eqref{c2000},
\begin{align}
	X_1=v \lfloor\dfrac{X_1}{v}\rfloor +[(X_4-X_2)\mod v].\label{c2005}
\end{align}
 Similarly,
\begin{align}
		X_4 &= (X_1+X_2-\lfloor\dfrac{X_1}{v}\rfloor ) \mod v,\label{c2006}\\
		X_2 &= (X_4-X_1+\lfloor\dfrac{X_1}{v}\rfloor ) \mod v\label{c2007}
\end{align}
Based on  \eqref{c2003}, \eqref{c2004}, \eqref{c2006} and \eqref{c2007},  the random vetor  $(X_i,i\in N_4)$ must satisfy the following equalities,
 \begin{align}
 	H(X_{123})&=H(X_{12})=H(X_{13}), \label{c2008} \\
 	H(X_{124})&=H(X_{12})=H(X_{14}).\label{c2009}
 \end{align}
 Note that 
\begin{align}
I(X_3;X_4|X_{12})=&H(X_{123})+H(X_{124})\nonumber\\&-H(X_{12})-H(X_{1234}).\label{c2014}
\end{align}
In light of \eqref{c2008} and \eqref{c2009}, replacing $H(X_{123})$ and $H(X_{124})$ by $H(X_{12})$ in \eqref{c2014},
\begin{align}
	I(X_3;X_4|X_{12})&=H(X_{12})-H(X_{1234})\\
&=- H(X_{34}|X_{12})
\end{align}
By the nonnegativity of $ H(X_{34}|X_{12})$ and
$I(X_3;X_4|X_{12})$, we have
\begin{equation}
  \label{eq:10}
  H(X_{1234})=H(X_{12}).
\end{equation}
Together with
\eqref{c2000}, which implies
\begin{equation}
  \label{eq:11}
  H(X_{1234})=H(X_{234}),
\end{equation}
and \eqref{c2011}, we have
\begin{align}
	H(X_{12})=3\log v. \label{c2012}
\end{align}
By \eqref{eq:10} and \eqref{c2008}, $3\log v=$
\rv{$H(X_{13})$}$ \le H(X_{134})\le
H(X_{1234})=3\log v$ which implies
\begin{equation}
  \label{eq:12}
  H(X_{134})=3\log v.
\end{equation}
 The mutual independence of $X_2$, $X_3$ and $X_4$ and \eqref{c2010} implies that 
 \begin{equation}
   \label{eq:13}
   H(X_{23})=H(X_{24})=H(X_{34})=2\log v
 \end{equation}
By \eqref{c2010}, \eqref{c2012}, \eqref{c2008}, \eqref{c2009},
\eqref{eq:13}, \eqref{eq:10}, \eqref{eq:11} and \eqref{eq:12}, the value of $H(X_A)$
are determined for all  $A \subseteq N$ except for $A=\{1\}$. Now we verify that $H(X_1)=2\log v$. 
 By \eqref{c2010} and \eqref{c2012},
 \begin{align}
 	H(X_1)\geq H(X_{12})-H(X_2)=2\log v.\label{c2013}
 \end{align}
On the other hand, as $X_1$ is distributed on  $\mathbb{I}_{v^2}$ by
\eqref{c2000}, we have $H(X_1)\leq 2\log v$. The proof is accomplished.
\end{proof}

\begin{theorem}
\label{vray}
  For $E=V^\alpha$, $a\h\in E$ is entropic if and only if $a=0$.
\end{theorem}
\begin{proof}
  For $a=0$, $a\h$ is the origin, and it is characterized by
  constant $(X_i,i\in N_4)$. For $a>0$, all $a\h$ dissatisfy
  Zhang-Yeung inequality, and are so non-entropic. 
\end{proof}

\section{Two-dimensional faces of $\Gamma_4$}
\label{enu}
 
As we discussed in the previous section, each face of a polyhedral
cone has two representations, i.e, H-representation, the intersection
of the facets containing the face, and V-representation, the convex
hull of the extreme rays contained in the face. For $2$-dimensional
faces of $\Gamma_4$, each of them contains exactly two extreme
rays. Thus in this paper we adopt the V-representation. For each two
extreme rays of $\Gamma_4$, to determine whether the convex hull of
them forms a face of $\Gamma_4$, we only need to check that the
intersection of the facets containing the two extreme rays is a face
of dimension $2$, that is, containing no other extreme rays. The
catalogue of all $2$-dimensional faces of $\Gamma_4$ is generated by
Algorithm \ref{alg:Framwork}. 


  \begin{algorithm}[H]
	\caption{ Two-dimensional Face of $\Gamma_4$ Generating Algorithm.}
	\label{alg:Framwork}
	\begin{algorithmic}[1]
		\Require
                The family $\mc{F}$ of all $28$ facets and the family $\mc{E}$ of all 41
                extreme rays of $\Gamma_4$.
		\Ensure
			Upper triangle of a $41\times 41$ $(0,1)$-matrix $C$,where
                        $C(i,j)=1$ if and only if the convex hull of the
                        $i$-th extreme ray $E_i$ and the $j$-th
                        extreme ray $E_j$ forms a $2$-dimensional
                        face of $\Gamma_4$.
	        \For{$1\le i<j\le 41$}
                     \State  $C(i,j) \gets 1$
                     \For{$k=1$ to $28$}
	                 \If{ the $k$-th facet $F_k$ contains both $E_i$ and $E_j$,}
                              \State put $F_k$ in $\mc{F}'$.
                         \EndIf
                      \EndFor
	               \For{ $E\in\mc{E}\setminus\{E_i, E_j\}$ }
	                           \If{  $E$ is contained in the face
                                     $\cap_{F\in\mc{F}'} F$}  
	                                  \State  $C(i,j) \gets 0$
	                                  \State \textbf{break}
	                           \EndIf
	                 \EndFor
	        \EndFor
	\end{algorithmic}
\end{algorithm}

The result of Algorithm \ref{alg:Framwork}, i.e., $510$ $2$-dimensional
faces corresponding to those ``$1$''-entries in
$C$ are classified in Table \ref{table:1}. 
Similar to the fact that the $41$ extreme rays can be classified in
$11$ types, the \rv{$510$}  $2$-dimensional faces can also be classified to
$59$ types, and in each type, each face can be obtained from another
by permuting the indices in $N_4$. These $59$ types of faces are
listed in the upper triangle of Table \ref{table:1}. We label the rows and
columns with the $11$ types of extreme rays of $\Gamma_4$.
For simplicity, we
denote the face $F=\mr{cone}(E_1, E_2)$ by $(E_1,E_2)$. In each cell
with ``$(E_1, E_2)\  n$,  Thm. $m$'', ``$(E_1, E_2)$'' denote a representative of the
type of $2$-dimensional faces, where ``$E_1$ $(E_2)$'' is a representative of
the type the extrem rays in the column (row),
``$n$'' is the number of the faces in this type and this face is
characterized in ``Thm. $m$'' in Section \ref{af}. For the cell with
``$0$'', the convex hull of the two extreme rays in each type forms no
$2$-dimensional faces of $\Gamma_4$.
More details
of the faces in Table \ref{table:1} are discussed in Section \ref{af}.

\begin{table*}[t] 
	\renewcommand{\arraystretch}{2}
	\begin{center}   		
			\huge
			\caption{Two-dimensional faces of $\Gamma_4$}  
			\label{table:1} 
			\resizebox{1\columnwidth}{!}{
					\begin{tabular}{|c|c|c|c|c|c|c|c|c|c|c|c|}  
							\hline   &$U^i_{1,1}$ & $U^\alpha_{1,2}$&$U^\alpha_{1,3}$ &$U_{1,4}$ &	$U^\alpha_{2,3}$&	$\mc{W}^\alpha_{2}$& 	$U_{2,4}$ &	$U_{3,4}$ &$\hat{U}^i_{2,5}$&	$\hat{U}^i_{3,5}$&	$V^\alpha_8$ \\
							\hline  \multirow{4}{*}{$U^j_{1,1}$ }     &    & $({U^{12}_{1,2},U^1_{1,1})}$&	$({U^{123}_{1,3},U^1_{1,1})}$ & &{$(U_{2,3}^{123},U_{1,1}^{1})$}  &	{$(\mc{W}_{2}^{14},U_{1,1}^{1})$}&	&	&	$(\hat{U}^1_{2,5},U_{1,1}^{1})$ &$(\hat{U}^{1}_{3,5},U^{1}_{1,1})$ &$(V^{12}_8,U^{1}_{1,1})$	\\
							&$({U^1_{1,1},U^2_{1,1})}$ &$12$,Thm.\ref{rk1}&$12$,                                                                                          Thm.\ref{rk1} & $({U_{1,4},U^{1}_{1,1})}$&{ $12$, Thm.\ref{rk3} } &{ $12$, Thm.\ref{rk3} } &{\color{gray}$(U_{2,4},U^{1}_{1,1})$} &$(U_{3,4},U^{1}_{1,1})$ &$4$,Thm.\ref{rk18} &$4$,Thm.\ref{rk21} &$12$,Thm.\ref{rk31}  \\
							\cline{3-4}\cline{6-7} \cline{10-12} 
							 &$6$, Thm.\ref{rk1} &$({U^{12}_{1,2},U^3_{1,1})}$  &$({U^{123}_{1,3},U^4_{1,1})}$ & $4$, Thm.\ref{rk1}& { $(U_{2,3}^{123},U_{1,1}^{4})$} &
							{$(\mc{W}_{2}^{34},U_{1,1}^{1})$} &{\color{gray}$4$} &$4$, Thm.\ref{rk8} & 	$(\hat{U}^1_{2,5},U_{1,1}^{2})$& 	$(\hat{U}^{1}_{3,5},U^{2}_{1,1})$ & $(V^{12}_8,U^{3}_{1,1})$\\
							& &$12$, Thm.\ref{rk1} &$4$, Thm.\ref{rk1} & &{ $4$, Thm.\ref{rk3}} &{ $12$, Thm.\ref{rk3}} & & &$12$, Thm.\ref{rk18} & $12$, Thm.\ref{rk21}&$12$, Thm.\ref{rk31} \\
							\hline  \multirow{6}{*}{$U^\beta_{1,2}$} &	\multirow{6}{*}{$\backslash$ }& & &	& &	{$(\mc{W}_{2}^{14},U_{1,2}^{14})$}&	& & & & \\
							& &$(U^{12}_{1,2},U^{13}_{1,2})$ &$({U^{123}_{1,3},U^{12}_{1,2})}$ & &{	$(U_{2,3}^{123},U_{1,2}^{12})$} &{$ 6$, Thm.\ref{rk3} 	} & & & & & \\
							\cline{7-7}
							& &$12$, Thm.\ref{rk1} &$12$, Thm.\ref{rk1} & $({U_{1,4},U^{12}_{1,2})}$&{ $12$, Thm.\ref{rk2}} & {$(\mc{W}_{2}^{24},U_{1,2}^{14})$}&{\color{gray}$(U_{2,4},U^{12}_{1,2})$} &$(U_{3,4},U^{12}_{1,2})$ &$(\hat{U}^1_{2,5},U_{1,2}^{12})$ &$(\hat{U}^{1}_{3,5},U^{12}_{1,2})$ & $(V^{12}_8,U^{13}_{1,2})$\\
							\cline{3-4} \cline{6-6}
							& &$(U^{12}_{1,2},U^{34}_{1,2})$ &$({U^{123}_{1,3},U^{14}_{1,2})}$ & $6$, Thm.\ref{rk1} &{ $(U_{2,3}^{123},U_{1,2}^{14})$} &{ $24$, Thm.\ref{rk3}} & {\color{gray}$6$} & $6$, Thm.\ref{rk9} & $12$, Thm.\ref{rk18} & $12$, Thm.\ref{rk21} & $24$, Thm.\ref{rk31} \\
							\cline{7-7}
							& &$3$, Thm.\ref{rk1} &$12$, Thm.\ref{rk1} & &{$12$, Thm.\ref{rk3}}  &{$(\mc{W}_{2}^{34},U_{1,2}^{12})$} & & & & & \\
							& & & & & & { $6$, Thm.\ref{rk2}}& & & & & \\
							\hline \multirow{2}{*}{$U^\beta_{1,3}$}  & \multirow{2}{*}{ $\backslash$} &	 \multirow{2}{*}{ $\backslash$} &	$({U^{123}_{1,3},U^{124}_{1,3})}$&	$({U_{1,4},U^{123}_{1,3})}$&	{$(U_{2,3}^{123},U_{1,3}^{124})$}&{	$(\mc{W}_{2}^{14},U_{1,3}^{124})$}&	{\color{gray}$(U_{2,4},U^{123}_{1,3})$}&	$(U_{3,4},U^{123}_{1,3})$&	$(\hat{U}^1_{2,5},U_{1,3}^{123})$&	$(\hat{U}^{1}_{3,5},U^{234}_{1,3})$&	$(V^{12}_8,U^{134}_{1,3})$\\
							 & & &$ 6$, Thm.\ref{rk1} &$ 4$, Thm.\ref{rk1} &{$ 12$, Thm.\ref{rk4}} & {$ 12$, Thm.\ref{rk2}} &{\color{gray}$4$} &$ 4$, Thm.\ref{rk11} &$ 12$, Thm.\ref{rk19} &$4$, Thm.\ref{rk22} &$12$, Thm.\ref{rk31} \\
							\hline \multirow{2}{*}{$U_{1,4}$ }        &	\multirow{2}{*}{$\backslash$} &	\multirow{2}{*}{$\backslash$} &	\multirow{2}{*}{$\backslash$}&	\multirow{2}{*}{$\backslash$}&	{	$(  U_{2,3}^{123  } , U_{1,4} )     $}&	\multirow{2}{*}{0}&	\multirow{2}{*}{0}&	$(U_{3,4},U_{1,4})$&\multirow{2}{*}{0}&	\multirow{2}{*}{0}&	$(V^{12}_8,U_{1,4})$\\
							& & & & &{$4$, Thm.\ref{rk7}} & & &$ 1$, Thm.\ref{rk12} & & &$ 6$, Thm.\ref{rk31} \\
							\hline \multirow{2}{*}{$U^\beta_{2,3}$}  &	\multirow{2}{*}{$\backslash$} &\multirow{2}{*}{	$\backslash$} &	\multirow{2}{*}{$\backslash$}&\multirow{2}{*}{	$\backslash$}&	$(  U_{2,3}^{123  } , U_{2,3}^{124  } )$ & {	$(  \mc{W}_{2}^{12  } , U_{2,3}^{134  } )$}&	{\color{gray}$(U_{2,4},U^{123}_{2,3}),4$}&	$(U_{3,4},U^{123}_{2,3})$&{\color{gray}$(\hat{U}^1_{2,5},U_{2,3}^{234})$}&	$(\hat{U}^{1}_{3,5},U^{123}_{2,3})$&$(V^{12}_8,U^{123}_{2,3})$\\
							& & & & &{ $6$, Thm.\ref{rk5}} &{ $12$, Thm.\ref{rk6}} & {\color{gray}$4$} &$ 4$, Thm.\ref{rk10} &{\color{gray}$4$} &$ 12$, Thm.\ref{rk23} & $ 12$, Thm.\ref{rk31}\\
							\hline \multirow{2}{*}{$\mc{W}^\beta_{2}$}	   &  \multirow{2}{*}{$\backslash$} &	\multirow{2}{*}{$\backslash$} &	\multirow{2}{*}{$\backslash$} &	\multirow{2}{*}{$\backslash$} &\multirow{2}{*}{$\backslash$}&{\color{gray}	$(\mc{W}_{2}^{12 },\mc{W}_{2}^{13 })$}&{	\color{gray}	$(U_{2,4},\mc{W}^{12}_{2})$}&	\multirow{2}{*}{0}&{\color{gray}	{$(\hat{U}^1_{2,5},\mc{W}_{2}^{12})$}}&{	$(\hat{U}^{1}_{3,5},\mc{W}^{23}_{2})$}&	\multirow{2}{*}{0} \\
						 	& & & & & &{\color{gray}$12$} &{\color{gray}$6$} & &{\color{gray}$12$} &$ 12$, Thm.\ref{rk24} &
						\\
							\hline  \multirow{2}{*}{$U_{2,4}$}       &\multirow{2}{*}{$\backslash$} &	\multirow{2}{*}{$\backslash$} &	\multirow{2}{*}{$\backslash$}&	\multirow{2}{*}{$\backslash$}&	\multirow{2}{*}{$\backslash$}&	\multirow{2}{*}{$\backslash$}& 	\multirow{2}{*}{$\backslash$}	& 	\multirow{2}{*}{0}&{\color{gray}$(\hat{U}^1_{2,5},U_{2,4})$}&{ \color{gray}	$(\hat{U}^{1}_{3,5},U_{2,4})$}&	\multirow{2}{*}{0}\\
							& & & & & & & & &$ 4$  &$ 4$  & \\
							\hline \multirow{2}{*}{$U_{3,4}$}     &  \multirow{2}{*}{$\backslash$} &\multirow{2}{*}{$\backslash$}&	\multirow{2}{*}{$\backslash$}&	\multirow{2}{*}{$\backslash$}&	\multirow{2}{*}{$\backslash$} &	\multirow{2}{*}{$\backslash$} &	\multirow{2}{*}{$\backslash$} &	\multirow{2}{*}{$\backslash$}	&	\multirow{2}{*}{0}&	\multirow{2}{*}{0}&	{$(V^{12}_8,U_{3,4})$}\\
						     	& & & & & & & & & & &$ 6$, Thm.\ref{rk31}  \\
							\hline $\hat{U}^j_{2,5}$      &	$\backslash$ &	$\backslash$ &	$\backslash$ &	$\backslash$&	$\backslash$&	$\backslash$&$\backslash$&	$\backslash$&	0&	0&	0\\
							\hline $\hat{U}^j_{3,5}$  &  $\backslash$ &$\backslash$&	$\backslash$&	$\backslash$&	$\backslash$&	$\backslash$&	$\backslash$&	$\backslash$&$\backslash$&	0&	0\\
							\hline  $V^\beta_8$    &  $\backslash$ &	$\backslash$ &	$\backslash$ &	$\backslash$&	$\backslash$&	$\backslash$&$\backslash$&$\backslash$&	$\backslash$&	$\backslash$&	0  \\
							\hline   
						\end{tabular}  
				}
		\end{center}   
              \end{table*}

\section{Characterization of entropy functions on two-dimensional faces of $\Gamma_4$ }
\label{af}

In this section, we characterize the entropy functions on
$2$-dimensional faces of $\Gamma_4$. We embed each face $F=(E_1, E_2)$
in the first octant of a $2$-dimensional cartesian coordinate system
whose axis are labeled by $a$ and $b$. Thus, for each $(a,b)$, $a,b\ge
0$, it represents the polymatroid $a\br_1+b\br_2$, where $\br_i, i=1,2$ is
the rank function of the minimal integer polymatroid in $E_i$, respectively. 
 \rv{Throughout}  this paper, for a random vector $(X_i,i\in N_4)$ or its subvectors, 
we assume each $X_i$ is distributed on  a finite set $\mathcal{X}_{i}$, and for each $x_i\in$ \rv{$\mathcal{X}_{i}$}, $p(x_i)>0$.
\begin{lemma}[Lemma 15.3, \cite{yeung2008information}]
  \label{lem}
  For any $\h_1,\h_2\in \Gamma^*_n$, $\h_1+\h_2\in \Gamma^*_n$.
\end{lemma}
\begin{lemma}
	\label{lem4}
	
	For a random vector $X_{i},i\in N_3$, consider the tripartite
        graph $G=(V,E)$, where $V=\mc{X}_1\cup \mc{X}_2\cup \mc{X}_3$
        and   $(x_i,x_j)\in E$   if $p(x_{i},x_{j})>0$ for distinct
        $i,j \in N_3 $. If $(X_i,i\in N_3)$ satisfies the following information equalities, 
	\begin{align}
		H(X_{ik})+H(X_{jk  })&=H(X_{ k})+H(X_{ijk } ), \nonumber 
	\end{align}
	where $i,j,k$ be any permutation of $1,2,3$, then each connected component of $G$ is a complete  tripartite graph. Futhermore, if $p(x_1)=p(x_2)=p(x_3)$ holds for any $p(x_1,x_2,x_3)>0$, then the number of vertices in $\mc{X}_i$, $i=1,2,3$ are the same  and the the probability mass of all of the vertices, the edges and the triangles  are equal, respectively, in each connected component. 
\end{lemma}

\begin{proof}
	For $ (x_i,i\in N_3) \in
	\mathcal{X}_{N_3}$ with $p(x_{123})>0$, the information equalities
	imply that the probability mass function satisfies
	\begin{align}
		p(x_1,x_2)p(x_1,x_3)&=p(x_1)p(x_1,x_2,x_3),  \label{e1} \\
		p(x_1,x_2)p(x_2,x_3)&=p(x_2)p(x_1,x_2,x_3),  \label{e2}  \\
		p(x_1,x_3)p(x_2,x_3)&=p(x_3)p(x_1,x_2,x_3) . \label{e3}
	\end{align}
To prove the lemma, for such $(X_i,i \in N_3)$, we now prove the following two claims.
	\begin{claim}
		\label{lem2}
		For any two vertices $x_1\in \mc{X}_1,x_2\in\mc{X}_2$, if they are  adjacent to the same vertex in $\mc{X}_3$, then these two vertices are also adjacent.
	\end{claim}
	\begin{proof}[\rv{Proof of Claim $\ref{lem2}$}]
		For any $x_1\in\mc{X}_1$, $x_2\in\mc{X}_2$  adjacent to the same vertex $x_3\in \mc{X}_3$, by definition of $G$, $p(x_1,x_3)$ and $p(x_2,x_3)$ are positive. By \eqref{e3}, we have $p(x_1,x_2,x_3)>0$ which implies $p(x_1,x_2)>0$  and so $x_1$ is adjacent to $x_2$. 
	\end{proof}
	\begin{claim}
		\label{lem3}
		For any two vertices  $x_1,x'_1\in\mc{X}_1$, if they are adjacent to a common vertex in $\mc{X}_2$ or $\mc{X}_3$, then each $y\in\mc{X}_2 $ or $\mc{X}_3$, either it is adjacent to both $x_1$ and $x'_1$, or it is adjacent to none of them. 
	\end{claim}
	\begin{proof}[\rv{Proof of Claim $\ref{lem3}$}]
		Assume there exists two vertices $x_1$, 
                $x_1'\in\mc{X}_1 $ with  $x_1\neq x_1'$ and they are adjacent to $x_2\in\mc{X}_2$, then $p(x_1,x_2)>0$ and $p(x_1',x_2)>0$. As
		\begin{equation}
			p(x_1,x_2)=\sum_{x_3}p(x_1,x_2,x_3) \nonumber
		\end{equation} 
		there exists $x_3\in\mc{X}_3$ such that $p(x_1,x_2,x_3)>0$. Thus both $x_1$ and $x_2$ are adjacent  to $x_3$. Since $x_1'$ is adjacent to $x_2$, by Claim \ref{lem2}, $x_1'$ is adjacent to $x_3$. For any vertex $y$ adjacent to $x_1'$, by Claim \ref{lem2}, $y$ is adjacent to $x_2$ or $x_3$. As $x_1$ is adjacent to $x_2$ and $x_3$, we conclude that $y$ is adjacent to $x_1$. By symmetry, $x_1'$ is adjacent to all vertices adjacent to $x_1$. 
	\end{proof}
	Now we prove that each connected component is a complete  tripartite graph. In a connected component, for any two vertices $x_i\in \mc{X}_i,x_j\in\mc{X}_j$, $i,j\in N_3,i\neq j$, there exists a path from $x_i$ to $x_j$ and  the length of the path is denoted by $n$. We prove that $x_i$ and $x_j$ are adjacent by induction on $n$. First, it is obviously ture for $n=1$ and it holds  for $n=2$ by Claim \ref{lem2}. Assume it is ture for some positive integers $k-1$ and $k$ where $k\geq 2$. Now consider $n=k+1$. If the vertex $y$ adjacent to $x_j$ on the path is not in $\mc{X}_i$, by induction hypothesis on $k$, we see that  $x_i$ is adjacent to $y$. Then by Lemma \ref{lem2}, $x_i$  is adjacent to $x_j$. If $y$ is in $\mc{X}_i$, the vertex $z(z\neq x_j)$ adjacent to $y$ on the path is adjacent to $x_i$ by induction on $k-1$. Then by Claim \ref{lem3}, $x_i$ is adjacent to $x_j$.
	
	For $p(x_1,x_2,x_3)>0$, if the probability mass function satisfies $p(x_1)=p(x_2)=p(x_3)$, then by \eqref{e1}-\eqref{e3},
	\begin{equation}
		p(x_1,x_2)=p(x_1,x_3)=(x_2,x_3) ,\label{e4}
	\end{equation}
	which implies that  the probability mass of vertices and edges are equal in each triangle. Since each connected component is a complete  tripartite graph, we conclude that the probability of vertex and edge are equal in each connected component. Then  by \eqref{e1}, the probability of the  triangles are all equal in each connected component. The proof is accomplished.
\end{proof}

      \subsection{All-entropic faces}
 \label{A}
\begin{theorem}
  \label{rk1}
  For $F=(E_1,E_2)$, where distinct $E_i, i=1,2$ contains a rank-$1$
  matroid, i.e., the first $4$ types of extreme rays, any $\h\in F$ is entropic.
\end{theorem}
\begin{proof}
  The fact that the whole $E_i$ is entropic, together with Lemma
  \ref{lem} immediately implies the theorem. 
\end{proof}

Theorem \ref{rk1} is also a corollary of \cite[Proposition 4]{chen2015marginal} or \cite[Theorem 2]{matuvs2016entropy}. 

\subsection{ Entropy functions on the faces extended from $2$-dimensional faces of $\Gamma_3$}
\label{B}
In this subsection, we characterize those faces which can be obtained from the characterizations of $2$-dimensional faces of $\Gamma_3$.
Thm.\ref{rk2} characterize  $3$ faces in Table \ref{table:1} which can be obtained by  the characterization of $2$-dimensional face
$(U_{2,3}, U^{12,3}_{1,2})$ of $\Gamma_3$. It has been done in \cite{matus2005piecewise}
. Thus, the entropy functions on
them have the same shape as those on $(U_{2,3}, U^{12,3}_{1,2})$,
which is depicted in Figure \ref{fig1}.
\begin{theorem}
  \label{rk2}
  For $F=(U_{2,3}^{123},U_{1,2}^{12}),
  (\mc{W}_{2}^{34},U_{1,2}^{12})$, \rv{and} $ (\mc{W}_{2}^{14},U_{1,3}^{124})$,  $\h=(a,b)\in F$ is entropic
  if and only if $a+b\ge \log \lceil 2^a \rceil$.\footnote{We take the
  base of the logarithm be $2$ for computing the entropy.}
\end{theorem}

\afterpage{
	\clearpage
	\begin{longtable}{p{1.7cm}|c|c|c|c}
		\caption{Entropy functions on two-dimensional faces of $\Gamma_4$ }
		\label{long} \\
		\hline Subsection & Theorem  &Two-dimensional faces&Entropy region  & Figure \\
		\hline
		\hline
		\endfirsthead
		\hline Subsection & Theorem & Two-dimensional faces&Entropy region & Figure\\
		\hline
		\hline
		\endhead
		\ref{A} &   Thm.\ref{rk1}   & \makecell{ \\ $(U^{1}_{1,1},U^{2}_{1,1})$, $(U^{12}_{1,2},U^{1}_{1,1})$, $(U^{12}_{1,2},U^{3}_{1,1})$, \\ $(U^{12}_{1,2},U^{13}_{1,2})$, $(U^{12}_{1,2},U^{34}_{1,2})$, $(U^{123}_{1,3},U^{1}_{1,1})$, \\ $(U^{123}_{1,3},U^{4}_{1,1})$, $(U^{123}_{1,3},U^{12}_{1,2})$, $(U^{123}_{1,3},U^{14}_{1,2})$, \\ $(U^{123}_{1,3},U^{124}_{1,3})$, $(U_{1,4},U^{1}_{1,1})$, $(U_{1,4},U^{12}_{1,2})$, \\ $(U_{1,4},U^{123}_{1,3})$.} &\makecell{ $\{a\br_1+b\br_2:$\\$a\geq 0, b\geq0\}$ }  & \makecell[l]{\\ \includegraphics[trim=10 10 0 10,scale=1]{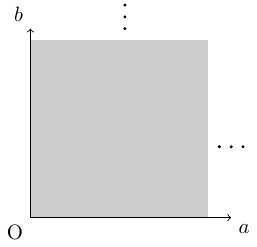}}  \\
		& &     & \\
		\hline
		\ref{B} & Thm.\ref{rk2} &\makecell{$(U^{123}_{2,3},U^{12}_{1,2})$,\\ $(\mc{W}_{2}^{34},U_{1,2}^{12})$,\\ $(\mc{W}_{2}^{14},U_{1,3}^{124})$. } &\makecell{ $\{a\br_1+b\br_2:$\\$ a+b\geq \log v   $ and \\ $ \log(v-1)\leq a \leq \log v$  \\
			for positive integer $ v$\} }& \makecell[l]{\includegraphics[trim=10 0 0 0,scale=1]{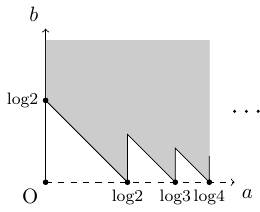}}
		\\
		\cline{2-5}
		&\makecell{\\  Thm.\ref{rk3}} & \makecell{ \\ $(U_{2,3}^{123},U_{1,1}^{1}), (U_{2,3}^{123},U_{1,1}^{4}),
			(U_{2,3}^{123},U_{1,2}^{14}), $\\ $  (\mc{W}_{2}^{14},U_{1,1}^{1}),
			(\mc{W}_{2}^{34},U_{1,1}^{1}), (\mc{W}_{2}^{14},U_{1,2}^{14})$,\\
			$(\mc{W}_{2}^{24},U_{1,2}^{14})$.}  &\makecell{$\{a\br_1+b\br_2:$\\$a=\log v$ for \\some positive \\integer $v$, $b\geq 0$\}  }& \makecell[l]{\includegraphics[trim=3 0 0 0,scale=1]{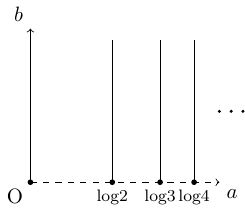}} \\
		\cline{2-5}
		&\makecell{\\  Thm.\ref{rk18}} &\makecell{ $(\hat{U}_{2,5}^{1},U_{1,1}^{1})$,\\ $(\hat{U}_{2,5}^{1},U_{1,1}^{2})$, \\ $(\hat{U}_{2,5}^{1},U_{1,2}^{12})$ }&\makecell{$\{a\br_1+b\br_2:$\\$a=\log v$ for \\some positive \\integer $v$, $b\geq 0$\}  }& \makecell[l]{\includegraphics[trim=3 0 0 0,scale=1]{fig_rk4.pdf}} 
		\\         
		\cline{2-5}
		& \makecell{  Thm.\ref{rk19}}& \makecell{$(\hat{U}_{2,5}^{1},U_{1,3}^{123})$} &\makecell{ $\{a\br_1+b\br_2:$\\$ a+b\geq \log v   $ and \\ $ \log(v-1)\leq a \leq \log v$  \\
			for positive integer $ v$\} }& \makecell[l]{\includegraphics[trim=10 0 0 0,scale=1]{fig_rk2.pdf}} \\ 
		\cline{2-5}
		& \makecell{  Thm.\ref{rk21}}& \makecell{$(\hat{U}_{3,5}^{1},U_{1,1}^{1})$, \\ $(\hat{U}_{3,5}^{1},U_{1,1}^{2})$, \\  $(\hat{U}_{3,5}^{1},U_{1,2}^{12})$} &\makecell{$\{a\br_1+b\br_2:$\\$a=\log v$ for \\some positive \\integer $v$, $b\geq 0$\}  }& \makecell[l]{\includegraphics[trim=3 0 0 0,scale=1]{fig_rk4.pdf}}\\
		\hline
		\multirow{20}{*}{ \makecell{\ref{C}} }& \makecell{ Thm.\ref{rk4}} &\makecell{ $( U_{2,3}^{123} , U_{1,3}^{124})$ }& \makecell{$\{a\br_1+b\br_2:$\\$a=\log v$ for \\some positive \\integer $v$, $b\geq 0$\}  }  &\makecell[l]{\includegraphics[trim=3 0 0 0,scale=1]{fig_rk4.pdf}} \\
		\cline{2-5}
		& \makecell{ Thm.\ref{rk5}} & \makecell{$(U_{2,3}^{123  } , U_{2,3}^{124  } ) $}& \makecell{$\{a\br_1+b\br_2:$\\$a=\log v_1, b=\log v_2$    \\for some positive\\ integer $v_1$, $v_2$\} }& \makecell[l]{\\ \includegraphics[trim=10 0 0 0,scale=1]{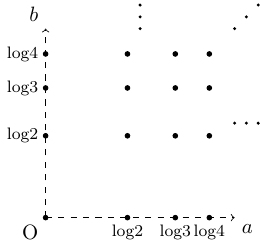}} \\
		\cline{2-5}
		& \makecell{ Thm.\ref{rk7}} & \makecell{$  ( U_{2,3}^{123} , U_{1,4}  ) $}  &\makecell{ $\{a\br_1+b\br_2:$\\$a\geq 0,b>0$ or\\ $(a,b)=(\log{v},0)$\\
			for  positive \\integer $k$\} } &  \makecell[l]{\\ \includegraphics[trim=5 0 0 0,scale=1]{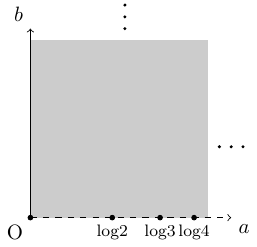}}\\
		\hline
		\multirow{1}{*}{ \makecell{\ref{D}} }& \makecell{ Thm.\ref{rk6}} &\makecell{$ ( \mc{W}_{2}^{12  } , U_{2,3}^{134  })$ }&\makecell{ $\{a\br_1+b\br_2:$\\$ a+b=\log{v}$,
			\\$a=H(\bm{\alpha})
			$, where \\ integer $v\ge 1$ and\\ $\bm{\alpha}$ is
			a  partition of $v$\}   } &  \makecell[l]{\\ \includegraphics[trim=10 0 0 0,scale=0.77]{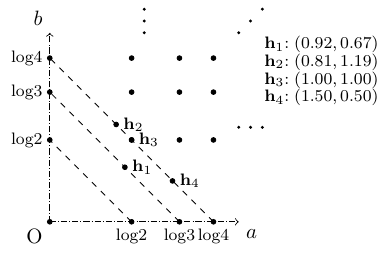}}\\
		\hline
		\ref{F} & \makecell{ Thm.\ref{rk8}} &$(  U_{3,4}  ,U^1_{1,1} ) $&\makecell{$\{a\br_1+b\br_2:$\\$a=\log v$ for \\some positive \\integer $v$, $b\geq 0$\}  } & \makecell[l]{\includegraphics[trim=3 0 0 0,scale=1]{fig_rk4.pdf}}  \\
		\hline
		&  \makecell{ Thm.\ref{rk9}} &$(  U_{3,4}  ,U^{12}_{1,2} )$ & \makecell{$\{a\br_1+b\br_2:$\\$a=\log v$ for \\some positive \\integer $v$, $b\geq 0$\}  } &\makecell[l]{\includegraphics[trim=3 0 0 0,scale=1]{fig_rk4.pdf}}  \\
		\cline{2-5}
		&  \makecell{ Thm.\ref{rk10}} &$(  U_{3,4}  ,U^{123}_{2,3} )$ & \makecell{$\{a\br_1+b\br_2:$\\$a+b=\log v$ for \\some positive\\ integer $v$\}  } & \makecell[l]{\\ \includegraphics[trim=10 0 0 0,scale=1]{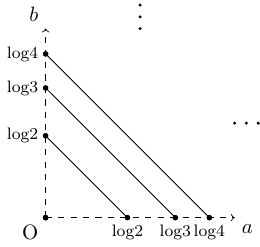}}  \\
		\cline{2-5}
		&  \makecell{ Thm.\ref{rk11}} &$(  U_{3,4}  ,U^{123}_{1,3} )$ &\makecell{$\{a\br_1+b\br_2:$\\$a=\log v$ for \\some positive\\ integer $v$, $b\geq 0$\}  } &\makecell[l]{\includegraphics[trim=3 0 0 0,scale=1]{fig_rk4.pdf}}  \\
		\cline{2-5}
		&  \makecell{ Thm.\ref{rk12}} &$(  U_{3,4}  ,U_{1,4} )$ &\makecell{$\{a\br_1+b\br_2:$\\$a\geq 0,b>0$ or\\ $(a,b)=(\log{v},0)$\\
			for  positive integer $v$\} } & \makecell[l]{\\ \includegraphics[trim=5 0 0 0,scale=1]{fig_rk7.pdf}}\\
		\hline
		\ref{G} &  \makecell{ Thm.\ref{rk22}} &$(\hat{U}_{3,5}^{1},U_{1,3}^{234})$ &\makecell{$\{a\br_1+b\br_2:$\\$a=\log v$ for \\some positive \\ integer $v$, $b\geq 0$\}  } & \makecell[l]{\includegraphics[trim=3 0 0 0,scale=1]{fig_rk4.pdf}}  \\
		\hline
		& \makecell{ Thm.\ref{rk23}} &$ (\hat{U}_{3,5}^{1},U_{2,3}^{123})$ &\makecell{$\{a\br_1+b\br_2:$\\$ a=\log v_1, b=\log v_2$\\ for some positive \\integer $v_1$, $v_2$\} }  & \makecell[l]{\\ \includegraphics[trim=0 0 0 0,scale=1]{fig_rk5.pdf}} \\
		\cline{2-5}
		& \makecell{ Thm.\ref{rk24}} &$(\hat{U}_{3,5}^{1},\mc{W}_{2}^{23})$ & \makecell{ $\{a\br_1+b\br_2:$\\$a+b=\log{v}$ for \\some integer $v\ge 1$\} }& \makecell[l]{\\ \includegraphics[trim=0 0 0 0,scale=1]{fig_rk10.pdf}}   \\
		\hline 
		\ref{J}  &  \makecell{ Thm.\ref{rk31}}&  \makecell{\\ $(V^{12}_8,  U^{1}_{1,1})$, $(V^{12}_8,  U^{3}_{1,1})$, $(V^{12}_8,  U^{12}_{1,2})$\\ $(V^{12}_8,  U^{134}_{1,3})$, $(V^{12}_8,  U_{1,4})$ \\ \quad}  &\makecell{ $\{a\br_1+b\br_2: $\\$ a= 0, b\geq 0$}  &\makecell[l]{\\ \includegraphics[trim=0 0 0 0,scale=1]{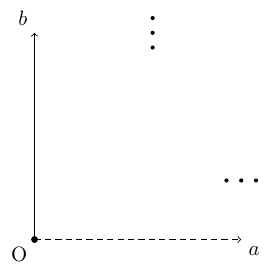}} \\
		\cline{2-5}
		&\makecell{ Thm.\ref{rk31}}&  \makecell{\\ $(V^{12}_8,  U^{123}_{2,3})$,\\ $(V^{12}_8,  U_{3,4})$ \\ \quad}  &\makecell{ $\{a\br_1+b\br_2: $\\$ a= 0, b=\log v$ for\\ some integer $v\ge 1\}$ }  &\makecell[l]{\\ \includegraphics[trim=0 0 0 0,scale=1]{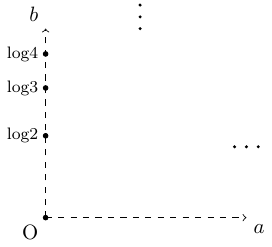}} \\
		\hline
	\end{longtable} 
}
\begin{proof}
Let $F$ be any of the $3$ faces and $\h\in F$. Let $\h=a\br_1+b\br_2$,
where $\br_1$ and $\br_2$ are the rank functions of the matroids on
the two extreme rays of the face, respectively. Restricting $\h$ on
$N_3$, we obtain $\h'=a\br'_1+b\br'_2$, where $\br'_i, i=1,2$ are the restriction
of $\br_i$ on $N_3$. It can be checked that they are 
the rank functions of $U_{2,3}$ and $U^{12,3}_{1,2}$,
respectively. Thus, $\h'\in (U_{2,3}, U^{12,3}_{1,2})$ and 
$a+b\ge \log \lceil 2^a \rceil$ by \cite[Theorem
1]{matus2005piecewise}. 

Now we prove the ``if '' part of the theorem. Let $\h=a\br_1+b\br_2\in
F$ and $a+b\ge \log \lceil
2^a \rceil$. Let $\h'$ be the restriction of $\h$ on $N_3$. Therefore,
$\h'\in  (U_{2,3}, U^{12,3}_{1,2})$ and let $(X_1,X_2,X_3)$ be its
characterizing random vector. Now for
\begin{itemize}
\item $F=(U_{2,3}^{123},U_{1,2}^{12})$, let $X_4$ be a
  constant;
\item $F=(\mc{W}_{2}^{34},U_{1,2}^{12})$, let $X_4=X_3$;
\item $F=(\mc{W}_{2}^{14},U_{1,3}^{124})$, let $X_4=X_1$.
\end{itemize}
Then $(X_i,i\in N_4)$ is the characterizing random vector of $\h$,
which implies that $\h$ is entropic.
\end{proof}
\vspace{-0.8cm}
%
%
\begin{figure}[H]
	\centering
	\includegraphics[scale=1]{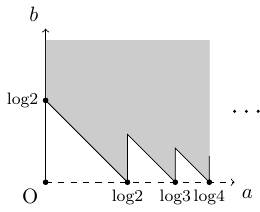}
	\captionsetup{justification=centering}
	\caption{The faces $(U_{2,3}^{123},U_{1,2}^{12})$, $ (\mc{W}_{2}^{34},U_{1,2}^{12})$ and $(\mc{W}_{2}^{14},U_{1,3}^{124})$}
	\label{fig1}
	
\end{figure}
For the $7$ faces characterized in Thm.\ref{rk3}, the entropy region on
them have the same shape as $2$-dimensional face $(U_{2,3},U^{1,3}_{1,1}
)$ of $\Gamma_3$, which has been characterized in \cite{chen2012characterizing}.


\begin{theorem}
  \label{rk3}
  For $F=(U_{2,3}^{123},U_{1,1}^{1})$, $(U_{2,3}^{123},U_{1,1}^{4})$, $
  (U_{2,3}^{123},U_{1,2}^{14})$, $  (\mc{W}_{2}^{14},U_{1,1}^{1})$, $
  (\mc{W}_{2}^{34},U_{1,1}^{1})$, $(\mc{W}_{2}^{14},U_{1,2}^{14})$ and
  $(\mc{W}_{2}^{24},U_{1,2}^{14})$, $\mathbf{h}=(a,b) \in F $ is entropic if
  and only if $a=\log{v}$ for integer $v\ge 1$.
\end{theorem}

\begin{figure}[H]
	\centering
     \includegraphics[scale=1]{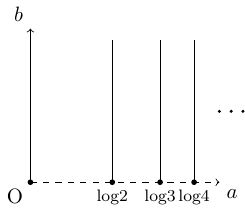}
	\captionsetup{justification=centering}
	\caption{The faces$(U_{2,3}^{123},U_{1,1}^{1}), (U_{2,3}^{123},U_{1,1}^{4}),
		(U_{2,3}^{123},U_{1,2}^{14}),   $ $  (\mc{W}_{2}^{14},U_{1,1}^{1}),
		(\mc{W}_{2}^{34},U_{1,1}^{1}), (\mc{W}_{2}^{14},U_{1,2}^{14})$ and
		$(\mc{W}_{2}^{24}, U_{1,2}^{14})$}
	\label{fig2}
\end{figure}

\begin{proof}
For $F=(U_{2,3}^{123},U_{1,1}^{4})$ and $\h\in F$,  restricting $\h$ on
$N_3$, we obtain $\h'=a\br'_1$, where $\br'_1$ is the rank function of
$U_{2,3}$, the matroid on  the extreme ray of $\Gamma_3$. Thus, if
$\h$ is entropic, then
$a=\log{v}$ for some integer $v> 0$.
Now for any $\h\in F$ with $a=\log{v}$, let $\h'$ be its restriction
on $N_3$ and $(X_1,X_2,X_3)$ be the
characterizing random variable of $\h'$. Let $X_4$ be an arbitrary
random variable independent of $(X_1,X_2,X_3)$ with $H(X_4)=b$. Then 
it can be checked that $(X_i,i\in N_4)$ is the characterizing
random variable of $\h$. Thus $\h$ is entropic.

For $F$ be any of the other $6$ faces and $\h\in F$, let $\h=a\br_1+b\br_2$,
where $\br_1$ and $\br_2$ are the rank functions of the matroids on
the two extreme rays of the face, respectively. Restricting $\h$ on
$N_3$, we obtain $\h'=a\br'_1+b\br'_2$, where $\br'_i=1,2$ is the restriction
of $\br_i$ on $N_3$, and  the rank function of $U_{2,3}$ and $U^{1,3}_{1,1}$,
respectively. Thus, $\h'\in (U_{2,3}, U^{1,3}_{1,1})$ and $a=\log{v}$
for positive $v\in\SetZ$
by \cite[Theorem
1]{chen2012characterizing}.

Now we prove the ``if '' part. Let $\h=a\br_1+b\br_2\in
F$ and $a=\log{v}$ for some integer $v\ge 0$. Let $\h'$ be the restriction of $\h$ on $N_3$. Therefore,
$\h'\in  (U_{2,3}, U^{1,3}_{1,1})$. Note that $\h'=a\br_1+b\br_2$,
where $\br_1$ and $\br_2$ are respectively the rank functions of $U_{2,3}$ and
$U^{1,3}_{1,2}$. Let $(X'_1,X'_2,X'_3)$ be the characterizing random
vector of $a\br_1$, that is, $X'_1$ and $X'_2$ are
independent and both uniformly distributed on $\mathbb{I}_v$, and
$X'_3=X'_1+X'_2\mod v$. Let $(X''_1,X''_2,X''_3)$ be characterizing random
vector of $b\br_2$, that is, $X''_1$ is an arbitrary random variable
with $H(X''_1)=b$, and $X''_2$ and $X''_3$ are both constant. 
Let $X_i=(X'_i,X''_i)$, $i\in N_3$. Then it can be checked that
$(X_1,X_2,X_3)$ is the characterizing random vector of $\h'$.

Now for
\begin{itemize}
\item $F=(U_{2,3}^{123},U_{1,1}^{1})$,  let $X_4$ be a
  constant;
\item $F=(U_{2,3}^{123},U_{1,2}^{14})$, let $X_4=X''_1$;
\item $F=(\mc{W}_{2}^{14},U_{1,1}^{1})$, let $X_4=X'_1$;
\item $F=(\mc{W}_{2}^{34},U_{1,1}^{1})$, let $X_4=X'_3$;
\item $F=(\mc{W}_{2}^{14},U_{1,2}^{14})$, let $X_4=X_1$;
\item $F=(\mc{W}_{2}^{24},U_{1,2}^{14})$, let $X_4=(X'_2,X''_1)$.
\end{itemize} 
Then $(X_i,i\in N_4)$ is the characterizing random vector of $\h$,
thus $\h$ is entropic.
\end{proof}


 

\begin{theorem}
	\label{rk18}
	For   $F=(\hat{U}_{2,5}^{1},U_{1,1}^{1})$, $(\hat{U}_{2,5}^{1},U_{1,1}^{2})$ and $(\hat{U}_{2,5}^{1},U_{1,2}^{12})$, $\h=(a,b)\in F$ is entropic if and only if $a=\log v$ for integer $v\ge 1$.
\end{theorem}
\begin{proof}
	Let $F$ be any of the $3$ faces and $\h\in F$. Let $\h=a\br_1+b\br_2$,
	where $\br_1$ and $\br_2$ are the rank functions of the minimal integer polymatroids on
	the two extreme rays of the face, respectively. Restricting $\h$ on
	$\{2,3,4\}$, we obtain $\h'=a\br'_1+b\br'_2$, where $\br'_i, i=1,2$ are the restriction
	of $\br_i$ on $\{2,3,4\}$. 
	
	By restriction, the above $3$ faces of $\Gamma_4$ are mapped to a subset of $\Gamma_3$ on $\{2,3,4\}$. The face $(\hat{U}_{2,5}^{1},U_{1,1}^{1})$ of $\Gamma_4$ is mapped to the extreme ray $U_{2,3}$ of $\Gamma_3$ on $\{2,3,4\}$, which implies that $a$ can only take $\log v$ for integer $v\ge 1$. 
	For the faces $(\hat{U}_{2,5}^{1},U_{1,1}^{2})$ and $(\hat{U}_{2,5}^{1},U_{1,2}^{12})$, by restriction, they are both mapped to the face  $(U_{2,3},U^{1}_{1,1})$ of $\Gamma_3$ on $\{2,3,4\}$, which implies that $a$ can only take $\log v$ for positive integer $v$ as well.  

	
	For the ``if'' part, as the whole ray $U_{1,1}^{i} $ and $U_{1,2}^{12}$ are entropic, by Lemma \ref{lem} and Theorem \ref{rk16}, all $a=\log v$ are entropic.
\end{proof}

\begin{theorem}
	\label{rk19}
	For   $F=(\hat{U}_{2,5}^{1},U_{1,3}^{123})$, $\h=(a,b)\in F$ is entropic if and only if $a+b\ge \log \lceil 2^a \rceil$.
\end{theorem}
\begin{proof}
	Let $\h=a\br_1+b\br_2$,
	where $\br_1$ and $\br_2$ are the rank functions of the minimal integer polymatroids on
	the two extreme rays of the face, respectively.
	Restricting $\h$ on
	$\{2,3,4\}$, we obtain $\h'=a\br'_1+b\br'_2$, where $\br'_i, i=1,2$ are the restriction
	of $\br_i$ on $\{2,3,4\}$.  By restriction, The face $(\hat{U}_{2,5}^{1},U_{1,3}^{123})$ of $\Gamma_4$ is mapped to the face $(U_{2,3},U^{12}_{1,2})$ of $\Gamma_3$ on $\{2,3,4\}$, Thus, $\h'\in (U_{2,3}, U^{12}_{1,2})$ and 
	$a+b\ge \log \lceil 2^a \rceil$ by \cite[Theorem
	1]{matus2005piecewise}.
	
	For the ``if '' part , we only need to show that $\h$ with $a+b=\log v$ is entropic for each positive integer $v$, then the theorem is implied by Lemma \ref{lem} and the fact that $b \br_2$ is entropic for all $b\geq0$. Let $X_2$ be uniformly distributed on  $\mathbb{I}_v$. Let $X_2$ and $X_4$ be independent, and $X_4$ be distributed on  $\mathbb{I}_v$ such that $H(X_4)=a$.
	 Let $X_1=vX_2+X_4$, $X_3=X_2+X_4 \mod v$. It can be checked that the entropy function of such constructed $(X_i,i\in N_4)$ is $(a,\log v-a)$.
\end{proof}

\begin{theorem}
	\label{rk21}
	For   $F=(\hat{U}_{3,5}^{1},U_{1,1}^{1})$, $(\hat{U}_{3,5}^{1},U_{1,1}^{2})$, or $(\hat{U}_{3,5}^{1},U_{1,2}^{12})$,  $\h=(a,b)\in F$ is entropic if and only if $a=\log v $ for some integer $v\ge 1$.
\end{theorem}
\begin{proof}
 If $\h'$  is entropic, where $\h'$ is obtained by restricting $\h=(a,b)$ on $\{1,3,4\}$, its characterizing  random vector $(X_i,i\in \{1,3,4\})$ satisfies 
	\begin{align}
		H(X_{134})&=H(X_{13})=H(X_{14}) \nonumber 
	\end{align}
and
\begin{align}
		H(X_{ij})&=H(X_i)+H(X_j)  ,\text{ distinct } i,j \in \{1,3,4\}\nonumber 
	\end{align}
 For $ x_{134}$  with $p(x_{134})>0$, above information equalities
	imply that the probability mass function satisfies
	\begin{align}
		p(x_1,x_3,x_4)&=p(x_1,x_3)=p(x_1,x_4)  \nonumber  \\
		&=p(x_1)p(x_3)=p(x_1)p(x_4),   \nonumber
	\end{align}
	which implies  $p(x_3)=p(x_4)$. Since $X_3$ and $X_4$ are independent, by Lemma \ref{lem1}, they are uniformly distributed on $\mc{X}_3$ and $\mc{X}_4$, respectively, and $H(X_3)=H(X_4)=\log v$, where $v=|\mc{X}_3|=|\mc{X}_4|$. As $(X_i, i\in \{1,3,4\})$ is
	its characterizing random vector, we have 
	\begin{align}
		H(X_3)&=H(X_4)=a.  \nonumber
	\end{align}
	Thus $a$ can only take $\log v$.
	
	 For the \rv{``if''} part, it can be proven by Lemma \ref{lem}  and the fact that $a=\log{v}$ in $\hat{U}_{3,5}^{1}$ and
	whole ray $U^{1}_{1,1}$, $U^{2}_{1,1} $ and $U_{1,2}^{12} $ are entropic.
\end{proof}

%

\subsection{Entropy functions on the faces involving $U_{2,3}^{123}$}
\label{C}
In this subsection, we characterize entropy functions on the $2$-dimensional faces of
$\Gamma_4$ with one of its extreme rays $U_{2,3}^{123}$ except those
characterized in Subsection \ref{B}. As introduced in \cite{CCB21},
uniform matroid $U_{2,3}$ corresponding to Latin squares or orthogonal
arrays $\mr{OA}(2,3)$, 
the characterizing random vectors of these
entropy functions are distributed on the
combinatorial structures
extending Latin squares. 


For $\h\in ( U_{2,3}^{123} , U_{1,3}^{124})$, if we restrict it on
$N_3$, we will see that its restriction $\h'\in (U_{2,3},
U^{12,3}_{1,2})$. However, we will prove in  Theorem \ref{rk4} that
the shape of the entropy functions on $( U_{2,3}^{123} ,
U_{1,3}^{124})$ is the same as $( U_{2,3} ,
U_{1,1}^{1,3})$, which is depicted in Figure \ref{fig6}. For $(
U_{2,3}^{123  } , U_{2,3}^{124  } )$ and $ ( U_{2,3}^{123} , U_{1,4}  ) $, they are
characterized in Theorems \ref{rk5} and \ref{rk7}, and
depicted in Figure \ref{fig3} and \ref{fig5}, respectively.

\begin{theorem}
	\label{rk4}
	For $F=( U_{2,3}^{123} , U_{1,3}^{124})$,  $\mathbf{h}=(a,b) \in
	F $ is entropic if and only if $a$= $\log{v}$ for integer $v\ge 1$.
\end{theorem}
\vspace{-0.5cm}
\begin{figure}[H]
	\centering
    \includegraphics[scale=1]{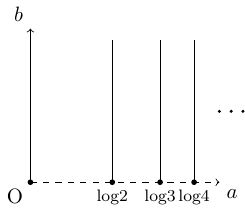}
	\captionsetup{justification=centering}
	\caption{The face $(U_{2,3}^{123} , U_{1,3}^{124})$ }
	\label{fig6}
\end{figure}
\begin{proof}
	If $\mathbf{ h} \in  F $ is entropic, its characterizing  random vector $(X_i,i\in N_4)$ satisfies the following information equalities,
	\begin{align}
	H(  X_{N_4})&=H(X_{N_4-i})  ,\ i\in N_4, \nonumber \\
	H(X_{3j })&=H(X_{3})+H(X_{j}),j\in N_4\setminus \{3\},  \nonumber \\
	H(X_{ ik})+H(X_{  jk})&=H(X_{ k})+H(X_{ijk } ), k\neq3, \nonumber\\ & \qquad\qquad\qquad\qquad\{i,j,k\}\neq\{1,2,3\}, \nonumber \\
	H(X_{ i\cup K})+H(X_{  j\cup K})&=H(X_{ K})+H(X_{  ij\cup K }), K\subseteq\{1,2,3\}\nonumber\\ & \qquad\qquad\quad  |K|=2, \{i,j\}= N_4\setminus K. \nonumber
\end{align}
  For $ (x_i,i\in N_4) \in
\mathcal{X}_{N_4}$ with $p(x_{1234})>0$, above information
equalities 
imply that the probability mass function satisfies
	\begin{align}
		p(x_{1},x_{2},x_{3},x_{4})&=p(x_{1},x_{2},x_{3}) \label{y1}\\
		&=p(x_{1},x_{2},x_{4}) \\ 
		&=p(x_{1},x_{3},x_{4})\\
		&= p(x_{2},x_{3},x_{4}), \label{y2}\\
		p(x_1,x_3)&=p(x_1)p(x_3), \label{y221}\\
		p(x_2,x_3)&=p(x_2)p(x_3), \label{y222}\\
		p(x_3,x_4)&=p(x_3)p(x_4), \\
		p(x_{1},x_{2})p(x_{1},x_{4})&=p(x_{1})p(x_{1},x_{2},x_{4}),\label{y3}\\
		p(x_{1},x_{2})p(x_{2},x_{4})&=p(x_{2})p(x_{1},x_{2},x_{4}),\label{y4}\\
		p(x_{1},x_{4})p(x_{2},x_{4})&=p(x_{4})p(x_{1},x_{2},x_{4}),\label{y5}\\
		p(x_{1},x_{3})p(x_{1},x_{4})&=p(x_{1})p(x_{1},x_{3},x_{4}),\label{y6}\\
		p(x_{1},x_{4})p(x_{3},x_{4})&=p(x_{4})p(x_{1},x_{3},x_{4}),\label{y7}\\  
		p(x_{2},x_{3})p(x_{2},x_{4})&=p(x_{2})p(x_{2},x_{3},x_{4}),\label{y8}\\
		p(x_{2},x_{4})p(x_{3},x_{4})&=p(x_{4})p(x_{2},x_{3},x_{4}),\label{y9}\\
		p(x_{1},x_{2},x_{3})p(x_{1},x_{2},x_{4})&=p(x_{1},x_{2})p(x_{1234}),\label{y10}\\
		p(x_{1},x_{2},x_{3})p(x_{1},x_{3},x_{4})&=p(x_{1},x_{3})p(x_{1234}),\label{y11}\\
		p(x_{1},x_{2},x_{3})p(x_{2},x_{3},x_{4})&=p(x_{2},x_{3})p(x_{1234}).\label{y12}
	\end{align}
	By \eqref{y1}-\eqref{y2} and \eqref{y10}-\eqref{y12},
	we have 
	\begin{align}
		p(x_{1},x_{2},x_{3},x_{4})&=p(x_{1},x_{2},x_{3})  \label{y13}\\
		&=p(x_{1},x_{2}) \label{y131}\\
		&=p(x_{1},x_{3}) \label{y132}\\
		&=p(x_{2},x_{3}),   \label{y14}
	\end{align}
	which implies that each of the random variables $X_i, i\in N_3$ is a
	function of the other two.
	
	Consider the bipartite graph $G=(V,E)$ with $V=\mc{X}_1\cup \mc{X}_2$
	and $(x_1,x_2)\in E$ if and only if $p(x_1,x_2)>0$. By \eqref{y131},
	for each $(x_1,x_2)\in E$, we color it by the unique
	$x_3\in \mc{X}_3$.
	Assume $|\mathcal{X}_{3}|=v$ and $G$ has $t$ connected components.
	\rv{We denote the number of the vertices of $\mc{X}_{i}$ in the connected component $C_{j}$ by $n_{i}^{(j)}, i=1,2, j=1,2,\cdots, t$  and  the  probability mass of $ C_{j}$ by $p_{j}$, that is, the  probability of the event that the random vector takes a tuple in $C_j$. }
	In
	light of \eqref{y132} and \eqref{y14}, equating \eqref{y221} and \eqref{y222},  we
	have 
	\begin{equation}
		\label{eq:3}
		p(x_1)=p(x_2).
	\end{equation}
	Let $(x'_1,x_2)\in E$ be another edge adjacent to $x_2$, by the same
	argument, we have
	\begin{equation}
		\label{eq:3}
		p(x_1')=p(x_2).
	\end{equation}
	As $x'_1$ is arbitrarily chosen, all vertices adjacent to $x_2$ have the same probability mass.
	Then by symmetry, so do all vertices in the
	connected component $C_j$, and $ n_{1}^{(j)}=n_{2}^{(j)}$ for each $j$. Herein, we simplify $n_{1}^{(j)}$ to $n^{(j)}$.
	By \eqref{y222}, $X_2$ and $X_3$ are independent, which implies that
	all colors $x_3\in\mc{X}_3$ appear in the 
	edges incident to a vertex $x_2\in \mc{X}_2$. Hence,
	\begin{equation}
		\label{eq:6}
		n^{(j)}\ge v
	\end{equation}
	for each $j$. Then,
	\begin{align}
		H(X_{1})&=-\sum_{i=1}^{t} n^{(i)}  \frac{p_{i}}{n^{(i)}}\log(\frac{p_{i}}{n^{(i)}}) \label{y18} \\
		&=H(p_{1},p_2,\ldots,p_t)+\sum_{i=1}^{t}p_{i}\log{n^{(i)}} .\label{y19}
	\end{align}
	
	By \eqref{y13}, canceling $p(x_{1},x_{2},x_{3})$ and
	$p(x_{1},x_{2},x_{3},x_{4})$ on either side of \eqref{y10}, we obtain
	\begin{equation}
		p(x_{1},x_{2})=p(x_{1},x_{2},x_{4}).\label{y17}
	\end{equation}
	Together with \eqref{y3} and \eqref{y4},
	we have 
	\begin{align}
		p(x_{1})&=p(x_{1},x_{4})\label{y15}
	\end{align}
	and
	\begin{align}
		p(x_{2})&=p(x_{2},x_{4})\label{y16}.
	\end{align}
	By \eqref{y17}, $X_{4}$ is a function of $X_{1}$ and $X_{2}$, and so
	we recolor each $(x_1,x_2)\in E$ by $x_{4}\in\mathcal{X}_{4}$. By
	\eqref{y15} and \eqref{y16}, all the edges adjacent to a vertex have
	the same color, which implies edges in the whole component have the
	same color. Hence
	\begin{align}
		H(X_{4})\leq H(p_{1},p_{2},\ldots,p_{t})\label{y23}.
	\end{align}
	The inequality holds strictly when there exist two components of $G$
	with the same color.
	
	As $\h\in ( U_{2,3}^{123} , U_{1,3}^{124})$ and $(X_i, i\in N_4)$ is
	its characterizing random vector, we have 
	\begin{align}
		H(X_{1})&=a+b ,\label{y20}\\
		H(X_{2})&=a+b,\\
		H(X_{3})&=a\label{y21}, \\
		H(X_{4})&=b \label{y22}.
	\end{align}
	
	Replacing the right hand side of \eqref{y20} into \eqref{y19} yields 
	\begin{align}
		a+b=H(p_1,p_2,\ldots,p_t)+\sum\limits_{i=1}^{t}p_{i}\log {n}^{(i)} .\label{y24}
	\end{align}
	By \eqref{y23} and \eqref{y22}, canceling $b$ and
	$H(p_1,p_2,\ldots,p_t)$ at either side of \eqref{y24}, we obtain
	\begin{align}
		a\ge \sum\limits_{i=1}^{t}p_{i}\log{n^{(i)}},
	\end{align}
	which together with \eqref{eq:6} yields
	\begin{equation}
		\label{eq:5}
		a\geq\sum\limits_{i=1}^{t}p_{i}\log{v}=\log{v}.
	\end{equation}
	
	Note that $H(X_{3})\leq \log{v}$. Thus, by \eqref{y21} and \eqref{eq:5},
	\begin{equation}
		\label{eq:7}
		\log{v}\geq a \geq \log{v},
	\end{equation}
	which implies that $a$ can only take $\log{v}$ for some positive
	integer $v$. 
	
	The ``if'' part  of the theorem is immediately implied by Lemma \ref{lem}  and the fact that
	 $a=\log{v}$ in $U_{2,3}^{123}$ and
	whole ray $U_{1,3}^{124}$ are entropic.
\end{proof}

\begin{theorem}
	\label{rk5}
	For $F=(U_{2,3}^{123  } , U_{2,3}^{124  } ) $, $\mathbf{h} =(a,b)\in
	F $ is entropic if and only if $a= \log{v_{1}}$ and $b= \log{v_{2}}  $
	for positive integers $v_{1},v_{2} $.
\end{theorem}
\begin{figure}[H]
	\centering
	  \includegraphics[scale=1]{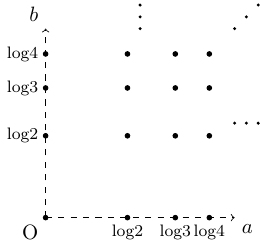}
	\caption{The face $  ( U_{2,3}^{123  } , U_{2,3}^{124  } ) $ }
	\label{fig3}
\end{figure}
 
\begin{proof}
	If $\mathbf{ h} \in  F $ is entropic, its characterizing random vector
	$(X_i,i\in N_4)$ satisfies the following information equalities,
	\begin{align}
	H(  X_{N_4})&=H(X_{N_4-i})  ,\ i\in N_4, \nonumber \\
	H(X_{ij })&=H(X_{i})+H(X_{j}),\ i,j\in N_4 ,i\neq j,   \nonumber\\
	H(X_{ ik })+H(X_{  jk })&=H(X_k)+H(X_{ ijk } ), \nonumber\\ &\qquad\qquad\{i,j,k\}\neq{ \{1,2,3\}, \{1,2,4\}},\nonumber\\
	H(X_{ 123 })+H(X_{ 124 })&=H(X_{  12 })  +H(X_{N_4}) .\nonumber
\end{align}
	  For $ (x_i,i\in N_4) \in
	\mathcal{X}_{N_4}$ with $p(x_{1234})>0$, above information equalities
	imply that the probability mass function satisfies
	\begin{align}
		p(x_{1},x_{2},x_{3},x_{4})&=p(x_{1},x_{2},x_{3}) \label{z1}\\
		&=p(x_{1},x_{2},x_{4}) \label{z17} \\ 
		&=p(x_{1},x_{3},x_{4}) \label{z18}\\
		&= p(x_{2},x_{3},x_{4}), \label{z2}\\
		p(x_1,x_2)&=p(x_1)p(x_2), \label{z3}\\
		p(x_1,x_3)&=p(x_1)p(x_3), \label{z4}\\
		p(x_1,x_4)&=p(x_1)p(x_4), \label{z5}\\
		p(x_2,x_3)&=p(x_2)p(x_3), \label{z6}\\
		p(x_2,x_4)&=p(x_2)p(x_4), \label{z7}\\
		p(x_3,x_4)&=p(x_3)p(x_4), \label{z8}\\
		p(x_{1},x_{3})p(x_{1},x_{4})&=p(x_{1})p(x_{1},x_{3},x_{4}),\label{z9}\\
		p(x_{1},x_{3})p(x_{3},x_{4})&=p(x_{3})p(x_{1},x_{3},x_{4}),\label{z10}\\
		p(x_{1},x_{4})p(x_{3},x_{4})&=p(x_{4})p(x_{1},x_{3},x_{4}),\label{z11}\\
		p(x_{2},x_{3})p(x_{2},x_{4})&=p(x_{2})p(x_{2},x_{3},x_{4}),\label{z12}\\
		p(x_{2},x_{3})p(x_{3},x_{4})&=p(x_{3})p(x_{2},x_{3},x_{4}),\label{z13}\\  
		p(x_{2},x_{4})p(x_{3},x_{4})&=p(x_{4})p(x_{2},x_{3},x_{4}),\label{z14}\\
		p(x_{2},x_{4})p(x_{3},x_{4})&=p(x_{4})p(x_{2},x_{3},x_{4}),\label{z15}\\
		p(x_{1},x_{2},x_{3})p(x_{1},x_{2},x_{4})&=p(x_{1},x_{2})p(x_{1234}).\label{z16}
	\end{align}
	By \eqref{z17}, canceling $p(x_{1},x_{2},x_{4})$ and $p(x_{1},x_{2},x_{3},x_{4})$ on either side of \eqref{z16}, we have
	\begin{equation}
		\label{z155}
		p(x_{1},x_{2},x_{3})=	p(x_{1},x_{2}).                                
	\end{equation}
	Equating \eqref{z3} and \eqref{z155} yields
	\begin{equation}
		p(x_{1},x_{2},x_{3})=p(x_{1})p(x_{2}).\label{a6}
	\end{equation}
	By \eqref{z4} and \eqref{z5},  replacing $p(x_{1},x_{3})$ and
	$p(x_{1},x_{4})$ by $p(x_1)p(x_3)$ and $p(x_1)p(x_4)$ in \eqref{z9},
	we have 
	\begin{equation}
		p(x_{1},x_{3},x_{4})=p(x_{1})p(x_{3})p(x_{4}).\label{a7}
	\end{equation}
	By the same argument, we have
	\begin{equation}  
		p(x_{2},x_{3},x_{4})=p(x_{2})p(x_{3})p(x_{4}).\label{a8}
	\end{equation}
	By \eqref{z1} and \eqref{z2}, equating \eqref{a6} and \eqref{a8} yields
	\begin{equation}                              
		p(x_{1})=p(x_{3})p(x_{4}).\label{a9}
	\end{equation}
	By \eqref{a7}, $X_{1}$, $X_{3}$ and $X_{4}$ are mutually
	independent. Let $x'_{3}\in \mathcal{X}_{3}$ with $x'_{3}\neq x_{3}$, then
	\begin{equation}
		p(x_{1},x'_{3},x_{4})=p(x_{1})p(x'_{3})p(x_{4})>0.
	\end{equation}
	As
	\begin{equation}
		p(x_{1},x'_{3},x_{4})=\sum_{x_{2}}p(x_{1},x_{2},x'_{3},x_{4}),
	\end{equation}
	there exists $x’_{2}\in \mathcal{X}_{2}$ such that
	\begin{equation}
		p(x_{1},x'_{2},x’_{3},x_{4})>0.
	\end{equation}
	By the same argument above, we obtain
	\begin{equation}
		p(x_{1})=p(x_{3}{'})p(x_{4}).  \label{a10}
	\end{equation}
	Equating \eqref{a9} and \eqref{a10}, we have
	\begin{equation}
		p(x_{3})=p(x_{3}{'}).\label{a1000}
	\end{equation}
	As \eqref{a1000} holds for any $x_{3}'\in \mathcal{X}_{3}$, $X_{3}$
	is uniformly distributed on $\mathcal{X}_{3}$, and so
	$H(X_3)=\log{v_1}$ where  $v_1=|\mathcal{X}_3|$. By the same argument, $X_4$ is uniform, and so 	$H(X_4)=\log{v_2}$ where  $v_2=|\mathcal{X}_4|$.
	
	For the ``if''  part of the theorem,  as $ \mathbf{h}\in U_{2,3}^{123}$
	and $ \mathbf{h'}\in U_{2,3}^{124}$ are entropic if $a=
	\log{v_{1}},b=\log{v_{2}}$ for integers $v_{1},v_{2}>0$ \cite{zhang1997non}. By Lemma \ref{lem}, all $(\log{v_{1}},\log{v_{2}} ) $ are entropic.
      \end{proof}
      
\begin{theorem}
	\label{rk7}
	For $F=(  U_{2,3}^{123  } , U_{1,4} )$, $\mathbf{h}=(a,b) \in F
	$ is entropic if and only if $a\geq 0,b>0$ or $(a,b)=(\log{v},0)$
	for  positive integer $v$. 
\end{theorem}
\begin{figure}[H]
	\centering
	\begin{tikzpicture}[scale = 2]
		\filldraw[gray!40,fill=gray!40] (0, 0) -- (1.5, 0) -- (1.5, 1.5) -- (0, 1.5);
		\draw [->][ black,densely dashdotted] (0,0)--(1.7,0) node[below right] { $a$};
		\draw [->] (0,0)--(0,1.6) node[above left] {$b$};
		\node[below left] at (0,0) {O};   
		\node[below,font=\fontsize{8}{6}\selectfont] at (0.693,0) {log2};  
		\node[below,font=\fontsize{8}{6}\selectfont] at (1.098,0) {log3};
		\node[below,font=\fontsize{8}{6}\selectfont] at (1.386,0) {log4};
		\fill (1.6,0.6) circle (0.4pt);  
		\fill (1.7,0.6) circle (0.4pt);  
		\fill (1.8,0.6) circle (0.4pt); 
		\fill (0.8,1.6) circle (0.4pt);  
		\fill (0.8,1.7) circle (0.4pt);  
		\fill (0.8,1.8) circle (0.4pt);
		\foreach \x in{0,0.693,1.098,1.386}
		\draw[fill] (\x,0) circle (.02);
	\end{tikzpicture}
	\caption{The face $  ( U_{2,3}^{123} , U_{1,4}  ) $ }
	\label{fig5}
\end{figure} 
\begin{proof}
	If $\mathbf{ h} \in  F $ is entropic, its characterizing
	random vector $(X_{i},i\in N_4)$ satisfies the following information equalities,
\begin{align}
	H(  X_{N_4})&=H(X_{N_4-i})  ,\ i\in N_4, \nonumber \\
	H(X_{  ik })+H(X_{ jk })&=H(X_{ k})+H(X_{  ijk } ), \{i,j,k\}\neq{ \{1,2,3\}}, \nonumber\\
	H(X_{i\cup K})+H(X_{  j\cup K})&=H(X_{ K})+H(X_{  ij
		\cup K }), K\subseteq \{1,2,3\} ,\nonumber\\
                    &\qquad\qquad\quad |K|=2,\{i,j\}= N_4\setminus K.\nonumber
\end{align}
 For $ \left(x_{1},x_{2},x_{3},x_{4}\right) \in \mathcal{X}_{N_4}$ with $p(x_{1},x_{2},x_{3},x_{4})>0$, the probability mass function satisfies 
	\begin{align}
		p(x_{1},x_{2},x_{3},x_{4})&=p(x_{1},x_{2},x_{3}) \label{z100} \\
		&=p(x_{1},x_{2},x_{4})  \label{a28}\\ 
		&=p(x_{1},x_{3},x_{4})  \label{z101}     \\
		&=p(x_{2},x_{3},x_{4}),\label{a29}  \\
		p(x_{1},x_{2})p(x_{1},x_{4})&=p(x_{1})p(x_{1},x_{2},x_{4}),\label{a30} \\
		p(x_{1},x_{2})p(x_{2},x_{4})&=p(x_{2})p(x_{1},x_{2},x_{4}),\label{a31} \\
		p(x_{1},x_{4})p(x_{2},x_{4})&=p(x_{4})p(x_{1},x_{2},x_{4}),\label{a32} \\
		p(x_{1},x_{3})p(x_{1},x_{4})&=p(x_{1})p(x_{1},x_{3},x_{4}),\label{a33} \\
		p(x_{1},x_{3})p(x_{3},x_{4})&=p(x_{3})p(x_{1},x_{3},x_{4}),\label{a34} \\
		p(x_{1},x_{4})p(x_{3},x_{4})&=p(x_{4})p(x_{1},x_{3},x_{4}),\label{a35} \\
		p(x_{2},x_{3})p(x_{2},x_{4})&=p(x_{2})p(x_{2},x_{3},x_{4}),\label{a36} \\
		p(x_{2},x_{3})p(x_{3},x_{4})&=p(x_{3})p(x_{2},x_{3},x_{4}),\label{a37} \\
		p(x_{2},x_{4})p(x_{3},x_{4})&=p(x_{4})p(x_{2},x_{3},x_{4}),\label{a38} \\
		p(x_{1},x_{2},x_{3})p(x_{1},x_{3},x_{4})&=p(x_{1},x_{2})p(x_{1234}),\label{a39} \\
		p(x_{1},x_{2},x_{3})p(x_{1},x_{3},x_{4})&=p(x_{1},x_{3})p(x_{1234}),\label{a40} \\
		p(x_{1},x_{2},x_{3})p(x_{2},x_{3},x_{4})&=p(x_{2},x_{3})p(x_{1234}).\label{a41}
	\end{align}
	By \eqref{z100}, canceling $p(x_{1},x_{2},x_{3})$ and $p(x_{1},x_{2},x_{3},x_{4})$ on either side of \eqref{a39}, we obtain
	\begin{equation}
		p(x_{1},x_{3},x_{4})=p(x_{1},x_{2}). \label{z102}
	\end{equation}
	Equating \eqref{z101} and \eqref{z102} yields,
	\begin{equation}
		p(x_{1},x_{2},x_{3},x_{4})=p(x_{1},x_{2}). \label{z103}
	\end{equation}
	By the same argument, 
	\begin{align}
		p(x_{1},x_{2},x_{3},x_{4})&=p(x_{1},x_{3}) \label{z104} \\
		&=p(x_{2},x_{3}) .\label{z105}
	\end{align}
	In light of \eqref{a28} and \eqref{z103}, we have 
	\begin{equation}
		p(x_{1},x_{2})=p(x_{1},x_{2},x_{4})\label{z106}.
	\end{equation}
	By \eqref{z106}, canceling $p(x_{1},x_{2})$ and  $p(x_{1},x_{2},x_{4})$ on either side of \eqref{a30}, we have 
	\begin{equation}
		p(x_{1},x_{4})=p(x_{1}).\label{b1}
	\end{equation}
	Similarly,
	\begin{align}
		p(x_{2},x_{4})&=p(x_{2}),\label{b2} \\
		p(x_{3},x_{4})&=p(x_{3}).\label{b3} 
	\end{align}
	By \eqref{b1}-\eqref{b3}, replacing $p(x_1,x_4)$ by $p(x_1)$, $p(x_2,x_4)$ by $p(x_2)$ and $p(x_3,x_4)$ by $p(x_3)$ in \eqref{a32}, \eqref{a35} and \eqref{a38}, we obtain
	\begin{align}
		p(x_1)p(x_2)=p(x_4)p(x_1,x_2,x_4) , \label{z107}\\
		p(x_1)p(x_3)=p(x_4)p(x_1,x_3,x_4)  ,\label{z108}\\
		p(x_2)p(x_3)=p(x_4)p(x_2,x_3,x_4)  .\label{z109}
	\end{align}
	Then, by \eqref{a28}-\eqref{a29}, \eqref{z107}-\eqref{z109} we conclude
	\begin{align}
		p(x_{1})=p(x_{2})=p(x_{3}) .\label{a43}
	\end{align}
	Consider the tripartite graph $G=(V,E)$ with $V=\mc{X}_1\cup \mc{X}_2\cup \mc{X}_3$ and   $(x_i,x_j)\in E$  if and only if $p(x_{i},x_{j})>0$, ${i,j}\in N_3 ,i\neq j$. By \eqref{z103}-\eqref{z105}, for each $(x_i,x_j)\in E$, we color it by the unique
	$x_4\in \mc{X}_4$. Assume $G$ has $t$ connected components.
	\rv{We denote the number of the vertices of $\mc{X}_{i}$ in the connected component $C_{j}$ by $n_{i}^{(j)}, i=1,2,3, j=1,2,\cdots,t$ } and  the  probability mass of $ C_{j}$ by $p_{j}$. 
	Now for any $(x_1,x_2)\in \mc{X}_1\times\mc{X}_2$ with $p(x_1,x_2)>0$, by
	\eqref{z103}, there exist unique $x_3$ and $x_4$ with
	$p(x_1,x_2,x_3,x_4)>0$, and so by \eqref{a43}, $p(x_1)=p(x_2)$. Similarly, for 
	$p(x_1,x_3)>0$  and $p(x_2,x_3)>0$, we have $p(x_1)=p(x_3)$
	and $p(x_2)=p(x_3)$. Then we conclude that  $p(x_i)$=$p(x_j)$
	when $p(x_i,x_j)>0$. Thus each vertex in a component of $G$ has the same probability mass. By \eqref{b1}-\eqref{b3}, all the edges adjacent to a vertex have
	the same color, which implies edges in the whole component have the
	same color. Hence
	\begin{equation}
		\label{eq:4}
		p(x_{4})\geq p_j.
	\end{equation}
	The inequality strictly holds when there exist two components
	sharing the same color.
	
	%
	
	By \eqref{a28}, replacing $p(x_{1},x_{2},x_{4})$ by $p(x_{1},x_{2},x_{3},x_{4})$ in  \eqref{z107}, we obtain
	\begin{equation}
		p(x_{1})p(x_{2})=p(x_{4})p(x_{1},x_{2},x_{3},x_{4}).\label{a45}
	\end{equation}
	By \eqref{a45}, it can be seen that $p(x_{1},x_{2},x_{3},x_{4})$ are
	all equal for $(x_{1},x_{2},x_{3},x_{4})\in\mathcal{X}^{(j)}_{N_4}$,
	where $\mathcal{X}^{(j)}_{N_4} $ is the alphabet $\mathcal{X}_{N_4}$
	in the connected component $ C_j$. Summing up  over all $p(x_1)$ in $C_j$,
	\begin{equation}
		p_j=\sum\limits_{x_1\in  \mathcal{X}^{(j)}_{1}} p(x_1) \label{z110}.
	\end{equation}
	Since each vertex has the same probability mass in $C_j$, the number of the vertices  $n_{1}^{(j)}=n_{2}^{(j)}=n_{3}^{(j)}=\dfrac{p_j}{p(x_{1})} $ where $p_j$ is the  probability mass of $ C_j$. We simplify $n_{1}^{(j)}$ to $n^{(j)}$, then $p(x_{1})=p(x_{2})=\dfrac{p_j}{n^{(j)}}$. Then substituting $p(x_1)$ and $p(x_2)$ by $ \dfrac{p_j}{n^{(j)}}$ in  \eqref{a45}, we have
	\begin{equation}
		p(x_{1},x_{2},x_{3},x_{4})=\dfrac{p_j^2}{({n^{(j)})}^2 p(x_{4})}.\label{a46}
	\end{equation}
	Summing up all probability mass in $C_j$,
	\begin{equation}
		p_j=\sum\limits_{(x_{1},x_{2},x_{3},x_{4})\in \mathcal{X}^{(j)}_{N_4}:\atop p(x_{1},x_{2},x_{3},x_{4})>0} p(x_{1},x_{2},x_{3},x_{4}).\label{a47}
	\end{equation}
	As $p(x_{1},x_{2},x_{3},x_{4})=p(x_{1},x_{2})$,
	\begin{align}
		p_j	=\sum\limits_{(x_{1},x_{2})\in \mathcal{X}^{(j)}_{N_2}:\atop p(x_{1},x_{2})>0} p(x_{1},x_{2}).\label{a48} 
	\end{align}
	Since $n_1^{(j)}=n_2^{(j)}=n^{(j)}$, we have 
	\begin{equation}
		| \{ x_{12}\in \mathcal{X}^{(j)}_{N_2} :p(x_{12})>0        \}|\leq  | \mathcal{X}^{(j)}_{N_2} |   =  ({n^{(j)})}^2 . \label{a49}
	\end{equation}
	Since $p(x_{1},x_{2},x_{3},x_{4})=p(x_{1},x_{2})$, $p(x_1,x_2)$ are equal for all $(x_1,x_2)\in \mathcal{X}^{(j)}_{N_2} $. Then by \eqref{a48} and \eqref{a49},
	\begin{equation}
		p_j\leq ({n^{(j)})}^2 p(x_{1},x_{2}).\label{a50}
	\end{equation}
	By \eqref{z103} , \eqref{a46} and \eqref{a50}, we have
	\begin{equation}
		\label{eq90}
		p_j\leq ({n^{(j)})}^2 p(x_{1},x_{2})=\dfrac{p_j^2}{ p(x_{4})},
	\end{equation}
	which implies $p_j\geq p(x_{4}).$ Together with \eqref{eq:4}, we have
	\begin{equation}
		\label{eq:9}
		p_j= p(x_{4}).
	\end{equation}
	Therefore, there exist no two components
	sharing the same color, that is
	\begin{equation}
		\label{eq:8}
		t=|\mc{X}_4|.
	\end{equation}
	By \eqref{eq:9}, \eqref{eq90} holds with equality which implies
	\eqref{a50} and \eqref{a49} hold with equality as well.
	Hence for any $x_1\in\mathcal{X}^{(j)}_1$,
	$x_2\in\mathcal{X}^{(j)}_2$, $x_1$ is adjacent to $x_2$. By the same
	argument, we can conclude that $x_i\in\mathcal{X}^{(j)}_i$ is adjacent
	to the vertex in $\mathcal{X}^{(j)}_{N_3-i}$ for $i=1,2,3$.

	Now we prove that for all $\h= (a,b)\in F$, $a\geq 0,b>0$ is
	entropic. \rv{For $a=\log v$ with some positive $v$ and $b>0$, the fact that $(a,b)$ is
        entropic can be implied by Lemma}\mrv{\ref{lem}}\rv{. Here we give a
        constructive proof of the case.}
      Let $X_{4}$ be any random variable on
	$\mathcal{X}_{4}=\{0,1,\ldots,t-1\}$ with $H(X_{4})=b$. Note
	that its probability mass is $p_i,i\in \mathcal{X}_{4}$. Thus 
	\begin{equation}
		b=\sum_{i=0}^{t-1}-p_{i}\log p_{i}. \label{a51}
	\end{equation} 
	Now assume in each connected component $C_j$ of $G$, $\mathcal{X}_{i}^{(j)}$ has $k$ element,
	that is, $n^{(j)}=k$ for all $j\in\mc{X}_4$.
	Let $X^{(j)}_{1}$ and $X ^{(j)}_{2}$ be independent and
	uniformly distributed on $\mathbb{I}_v$, and $X^
	{(j)}_{3}=X^ {(j)}_{1}+X^ {(j)}_{2} \mod v$.
	Let $X_i=(X^{(j)}_i,j\in\mc{X}_4)$, $i=1,2,3$.
	Then 
	\begin{align}
		H(X_{1})&=H(X_{2})=H(X_{3})=\sum_{j=0}^{t-1}
		-p_{j}\log \frac{p_{j}}{v}\\
		&= \sum\limits_{j=0}^{t-1}p_{j}\log{v}+\mrv{\sum_{j=0}^{t-1}} -p_{j}\log p_{j}\\
		&=\log{v}+H(X_{4})=a+b.
	\end{align}
	It can be checked that the entropy function of such
	constructed $(X_i, i\in N_4)$ is $(a,b)$.
	
	Consider $a\neq\log{v},b>0$. Let $c=\lfloor
2^a\rfloor$. Let $X_{4}$ be a random variable on
$\mathcal{X}_{4}=\{0,1,\ldots,t-1\}$ with $H(X_{4})=b$. 
Let $n^{(j)}=c$ for $j\in\mc{X}_4\setminus\{t-1\}$ and $n^{(t-1)}=c+l$ for
some sufficiently large $l$. 
Let 
$X^{(j)}_{1}$ and $X^{(j)}_{2}$ be independent and uniformly distributed on
$\{0,1,\dots,c-1\}$ for $j\in\mc{X}_4\setminus\{t-1\}$ and $X^{(j)}_{3}=X^{(j)}_{1}+X^{(j)}_{2}
\mod c$. Let $X^{(t-1)}_{1}$ and $X^{(t-1)}_{2}$ be independent
and uniformly distributed on $\{0,1,\dots,c+l-1\}$ and
$X^{(t-1)}_{3}=X^{(t-1)}_{1}+X^{(t-1)}_{2} \mod c+l$.  Let $X_i=(X^{(j)}_i,j\in\mc{X}_4)$, $i=1,2,3$.
Then,
\begin{align}
	H(X_{1})&=H(X_{2})=H(X_{3})\\
	&=\sum\limits_{i=0}^{t-2}(-p_{i}\log(\dfrac{p_{i}}{c})) -p_{t-1}\log\dfrac{p_{t-1}}{c+l}\\
	&=\sum\limits_{i=0}^{t-1}-p_{i}\log(p_{i})+\sum\limits_{i=0}^{t-2}p_{i}\log c+p_{t-1}\log(c+l)\\
	&=H(X_{4})+\log{c}({1-p_{t-1}})+\log{(c+l)}p_{t-1}.
\end{align}
	Note that $l$ is sufficiently large, so let
	\begin{align}
		p_{t}=\dfrac{a-\log{c}}{\log{(c+l)}-\log{c}}, \label{a52}
	\end{align}
	we have  
	\begin{align}
		H(X_{1})&=H(X_{2})=H(X_{3})=H(X_{4})+a.
	\end{align}
	It can be checked that the entropic function of such
	constructed $(X_i,i\in N_4)$ is $(a,b)$.
	
	Hence, all $\h=(a,b)\in F$ are entropic when
	$a\geq0,b>0$. Note that $ \mathbf{h}\in U_{2,3}^{123}$
	are entropic if and only if $a=\log{v}$ for positive
	integer $v$. The theorem is proved.
\end{proof}

Theorem \ref{rk7} can be considered as a corollary of \cite[Lemma 4]{matuvs2007},
which is restated as \cite[Lemma 6]{csirmaz2025exp}.

\subsection{Entropy functions on the faces involving partitions}
\label{D}
 In this subsection, we characterize the faces where the distribution of the entropy function’s characterizing random vector involves \emph{number partitions}. A $t$-partition of a positive $v$ with $1\leq t \leq v$ is a $t$-tuple \(  \bm{\alpha}=(\alpha_0, \alpha_1, \dots, \alpha_{t-1})\) with $\alpha_i$ positive integer and \(\sum_{i=0}^{t-1} \alpha_i = v\). The family of all such partitions of \( v \) is denoted by \(\mathcal{P}(v)\) and  $H(\bm{\alpha})$ reprensent the shannon entropy $H( \frac{\alpha_i}{v}:i\in \mathbb{I}_t)$.
 
 
\begin{theorem}
	\label{rk6}
	For $F=( \mc{W}_{2}^{12  } , U_{2,3}^{134  }) $, $\mathbf{h}=(a,b) \in F $
	is entropic if and only if $ a+b=\log{v},
	a=H(\dfrac{\alpha_{0}}{v},\dfrac{\alpha_{1}}{v},\dots,\dfrac{\alpha_{t-1}}{v})
	$ , where integer $v\ge 1$ and $(\alpha_0,\alpha_1,\dots,\alpha_{t-1})$ is
	a  partition of $v$.
\end{theorem}

\begin{figure}[H]
	\centering
	\begin{tikzpicture}[scale = 2]
		\draw [->,densely dashdotted] (0,0)--(1.6,0) node[below right] { $a$};
		\draw [->,densely dashdotted] (0,0)--(0,1.6) node[above left] {$b$};
		\draw [black, dashed,thin] (0.693,0)--(0,0.693);
		\draw [black, dashed,thin] (1.0985,0)--(0,1.098);
		\draw [black, dashed,thin] (1.386,0)--(0,1.386);
		\fill (1.6,0.8) circle (0.4pt);  
		\fill (1.7,0.8) circle (0.4pt);  
		\fill (1.8,0.8) circle (0.4pt);  
		\fill (0.8,1.6) circle (0.4pt);  
		\fill (0.8,1.7) circle (0.4pt);  
		\fill (0.8,1.8) circle (0.4pt);
		\fill (1.6,1.6) circle (0.4pt);  
		\fill (1.7,1.7) circle (0.4pt);  
		\fill (1.8,1.8) circle (0.4pt); 
		\draw[fill] (0.636,0.462) circle (.02);
		\node[right,font=\fontsize{8}{6}\selectfont] at(0.636,0.462){\rv{$\h_1$}}; 
		\draw[fill] (0.562,0.824) circle (.02);
			\node[right,font=\fontsize{8}{6}\selectfont] at(0.562,0.824) {\rv{$\h_2$}}; 
		\draw[fill] (0.693,0.693) circle (.02);
			\node[right,font=\fontsize{8}{6}\selectfont] at(0.693,0.693){\rv{$\h_3$}}; 
		\draw[fill] (1.039,0.346) circle (.02);
			\node[right,font=\fontsize{8}{6}\selectfont] at(1.039,0.346){\rv{$\h_4$}}; 
		\node[below left] at (0,0) {O};   
		\foreach \x in{0,0.693,1.098,1.386}
		\foreach \y in{0,0.693,1.098,1.386}
		\draw[fill] (\x,\y) circle (.02);
		\node[below,font=\fontsize{8}{6}\selectfont] at (0.693,0) {log2};  
		\node[below,font=\fontsize{8}{6}\selectfont] at (1.098,0) {log3};
		\node[below,font=\fontsize{8}{6}\selectfont] at (1.386,0) {log4};
		\node[left,font=\fontsize{8}{6}\selectfont] at (0,0.693) {log2};
		\node[left,font=\fontsize{8}{6}\selectfont] at (0,1.098) {log3};
		\node[left,font=\fontsize{8}{6}\selectfont] at (0,1.386) {log4};
		 \node[font=\fontsize{8}{6}\selectfont] at (2.3,1.5){\rv{$\h_1$}$:(0.92,0.67)$};
		\node[font=\fontsize{8}{6}\selectfont] at (2.3,1.35){\rv{$\h_2$}$:(0.81,1.19)$};
		\node[font=\fontsize{8}{6}\selectfont] at (2.3,1.2){\rv{$\h_3$}$:(1.00,1.00)$};
		\node[font=\fontsize{8}{6}\selectfont] at (2.3,1.05){\rv{$\h_4$}$:(1.50,0.50)$};
	\end{tikzpicture}	
	\caption{The face $  ( \mc{W}_{2}^{12  } , U_{2,3}^{134  })  $ }
	\label{fig4}
\end{figure}

\begin{proof}
	If $\mathbf{ h} \in  F $ is entropic, its characterizing
	random vector $(X_{i},i\in N_4)$ satisfies the following information equalities,
\begin{align}
	H(  X_{N_4})&=H(X_{N_4-i})  ,\ i\in N_4 \nonumber \\
	H(X_{ij })&=H(X_{i})+H(X_{j}),\{i,j\}\neq\{1,2\} \nonumber \\
	H(X_{  ik })+H(X_{jk })&=H(X_{ k})+H(X_{  i,j,k} ), k\in \{1,2\}\subseteq{ \{i,j,k\}},\nonumber \\
	H(X_{ i\cup K})+H(X_{  j\cup K})&=H(X_{ K})+H(X_{  ij\cup K }),K\subseteq \{1,3,4\}, \nonumber\\ &\qquad\qquad \qquad |K|=2,\{i,j\}=N_4\setminus K.\nonumber 
\end{align}
  For $ (x_i,i\in N_4) \in
	\mathcal{X}_{N_4}$ with $p(x_{1234})>0$, above information equalities
	imply that the probability mass function satisfies
	\begin{align}
		p(x_{1},x_{2},x_{3},x_{4})&=p(x_{1},x_{2},x_{3}) \label{z19}\\
		&=p(x_{1},x_{2},x_{4}) \label{z20} \\ 
		&=p(x_{1},x_{3},x_{4}) \label{z21}\\
		&= p(x_{2},x_{3},x_{4}), \label{z22}\\
		p(x_1,x_3)&=p(x_1)p(x_3), \label{z23}\\
		p(x_1,x_4)&=p(x_1)p(x_4), \label{z24}\\
		p(x_2,x_3)&=p(x_2)p(x_3), \label{z25}\\
		p(x_2,x_4)&=p(x_2)p(x_4), \label{z26}\\
		p(x_3,x_4)&=p(x_4)p(x_4), \label{z27}\\
		p(x_{1},x_{2})p(x_{1},x_{3})&=p(x_{1})p(x_{1},x_{2},x_{3}), \label{a14} \\
		p(x_{1},x_{2})p(x_{2},x_{3})&=p(x_{2})p(x_{1},x_{2},x_{3}), \label{z28} \\
		p(x_{1},x_{2})p(x_{1},x_{4})&=p(x_{1})p(x_{1},x_{2},x_{4}), \label{z29} \\
		p(x_{1},x_{2})p(x_{2},x_{4})&=p(x_{2})p(x_{1},x_{2},x_{4}), \label{z30} \\
		p(x_{1},x_{2},x_{3})p(x_{1},x_{3},x_{4})&=p(x_{1},x_{3})p(x_{1234}),\label{z31}\\
		p(x_{1},x_{2},x_{4})p(x_{1},x_{3},x_{4})&=p(x_{1},x_{4})p(x_{1234}),\label{z32}\\
		p(x_{2},x_{3},x_{4})p(x_{1},x_{3},x_{4})& =p(x_{3},x_{4})p(x_{1234}).\label{z33}
	\end{align}
	By \eqref{z19}, canceling $p(x_{1},x_{2},x_{3})$ and
	$p(x_{1},x_{2},x_{3},x_{4})$ on either side of \eqref{z31} yields
	\begin{equation}
		p(x_{1},x_{3},x_{4})=p(x_{1},x_{3}).\label{z34}
	\end{equation} 
	Together with \eqref{z23}, we have
	\begin{equation}
		p(x_{1},x_{3},x_{4})=p(x_{1})(x_{3}). \label{a11}
	\end{equation} 
	By the same argument, 
	\begin{align}
		p(x_{1},x_{3},x_{4})&=p(x_{1},x_{4})=p(x_{1})p(x_{4}) \label{a12}  \\  
		&=p(x_{3},x_{4})=p(x_{3})p(x_{4}). \label{a13} 
	\end{align}
	Equating \eqref{a11}, \eqref{a12} and \eqref{a13}, we have
	\begin{align}
		p(x_{1})=p(x_{3})=p(x_{4}).\label{z37}
	\end{align}
	By \eqref{z34}-\eqref{a13}, $X_{1},X_{3}$ and $X_{4}$ are  pairwise
	independent and each one is a function of the other two. Thus,
	$(X_1,X_3,X_4)$ is the characterizing random vector of an entropy
	function on $U_{2,3}$, an extreme ray of $\Gamma_3$, and $X_{i}$ is
	uniformly distributed on $\mathcal{X}_{i}$ and $H(X_i)=\log v$ where  $|\mathcal{X}_{i}|=v$ for $i=1,3,4$.

	As $\h\in F$, $(X_i, i\in N_4)$ is
	its characterizing random vector, we have 
	\begin{align}
		H\left(X_{1}\right)&=a+b=\log{v}, \label{d1} \\
		H(X_{2})&=a ,\label{d2} \\
		H\left(X_{3}\right)&=a+b=\log{v} ,\label{d3} \\
		H\left(X_{4}\right)&=a+b=\log{v} .\label{d4} 
	\end{align}
	By \eqref{z19}, \eqref{z21} and \eqref{z34},
	\begin{equation}
		p(x_{1},x_{2},x_{3})=p(x_{1},x_{3}).\label{z35}
	\end{equation}
	By \eqref{z35}, canceling $p(x_{1},x_{2},x_{3})$ and $p(x_{1},x_{3})$  on either side of \eqref{a14}, we have
	\begin{align}
		p(x_{1})=p(x_{1},x_{2}).\label{z36}
	\end{align}
	By \eqref{z36}, $X_2 $ is a function of $X_1$. Thus there exists a unique $ x_{2}\in\mathcal{X}_{2}$ with $ p(x_{1},x_{2})>0$, then  $|\mathcal{X}_{2}|\leq|\mathcal{X}_{1}|=v$ and
	\begin{align}
		p(x_{2})=\sum\limits_{x_{1}} p(x_{1},x_{2})=\dfrac{\alpha \left(x_{2}\right)  }{v},
	\end{align}
	where $\alpha (x_{2}) $ is the number of $ x_{1}$ with $p(x_{1},x_{2})>0$. Then $( \alpha \left(x_{2}\right): x_{2}\in  \mathcal{X}_{2}  ) $ is a number partition of $v$, i.e.,$\sum\limits_{x_{2}\in \mathcal{X}_{2} }\alpha(x_{2})=v$. By \eqref{d2}, $a$ can only take the vaule $ H(\dfrac{\alpha (x_{2})}{v} : x_{2}\in  \mathcal{X}_{2}   ) $.
	
	Now we prove all these points are entropic. Let $X_{1},X_{3}$ be indenpendent and uniformly distributed on $\mathcal{X}_{1}=\mathcal{X}_{3}=\mathbb{I}_v$. Let $X_{4}=X_{1}+X_{3}$ mod $v$. It can be checked that $X_{4}$ is uniformly distributed on $\mathcal{X}_{4}=\mathbb{I}_v $. Let $\left(A(x_{2}):x_{2}\in\mathcal{X}_{2} \right)$ be a partition of $\mathbb{I}_v $, where $| A(x_{2})|=\alpha(x_{2}).$ Let $X_{2}$ be a function $\varphi$ of $X_{1}$, where $\varphi(x_{1})=x_{2}$ if  $x_{1}\in A(x_{2}) $. Then it can be checked that $  (X_{i},i\in N_4)  $ is the characterizing
	random vector of $\h \in F$.
\end{proof}

\subsection{Entropy functions on the faces involving $U_{3,4}$}
\label{F}

In this subsection, we characterize entropy functions on the $5$
$2$-dimensional faces of
$\Gamma_4$ with one of its extreme rays $U_{3,4}$. As introduced in \cite{CCB21},
uniform matroid $U_{3,4}$ corresponding to Latin cubes or orthogonal
arrays $\mr{OA}(3,4)$, the characterizing random vectors of these
entropy functions are distributed on the
combinatorial structures
extending Latin cubes. 



 \begin{theorem}
	\label{rk8}
	For $F=(  U_{3,4}  ,U^1_{1,1} )$, $\mathbf{h}=(a,b) \in F
	$ is entropic if and only if  $a$= $\log{v}$ for integer $v\ge 1$ . 
\end{theorem}

 \begin{proof}
	If $\mathbf{ h} \in  F $ is entropic, its characterizing random vector
	$(X_i,i\in N_4)$ satisfies the following information equalities,
	\begin{align}
		H(  X_{N_4})&=H(X_{N_4-i})  ,\ i\in N_4\setminus\{1\}, \nonumber \\
		H(X_{ij })&=H(X_{i})+H(X_{j}) ,\ i,j\in N_4 ,i<j,   \nonumber \\
		H(X_{ ik})+H(X_{  jk})&=H(X_{ k})+H(X_{ijk } ), \{i,j,k\} \subseteq  N_4.\nonumber 
	\end{align}
	 For $ (x_i,i\in N_4) \in
	\mathcal{X}_{N_4}$ with $p(x_{1234})>0$, above information equalities
	imply that the probability mass function satisfies
	\begin{align}
		p(x_{1},x_{2},x_{3},x_{4})&=p(x_{1},x_{2},x_{3}) \label{c1} \\
		&=p(x_{1},x_{2},x_{4})  \label{c2}\\ 
		&=p(x_{1},x_{3},x_{4}) , \label{c3}   \\
		p(x_1,x_2)&=p(x_1)p(x_2),   \label{c17}   \\
		p(x_1,x_3)&=p(x_1)p(x_3) ,  \label{c18}   \\
		p(x_1,x_4)&=p(x_1)p(x_4)  , \label{c19}   \\
		p(x_2,x_3)&=p(x_2)p(x_3) ,  \label{c20}   \\
		p(x_2,x_4)&=p(x_2)p(x_4) ,  \label{c21}   \\
		p(x_3,x_4)&=p(x_3)p(x_4) ,  \label{c22}   \\		
		p(x_{1},x_{2})p(x_{1},x_{3})&=p(x_{1})p(x_{1},x_{2},x_{3}),\label{c5} \\
		p(x_{1},x_{2})p(x_{2},x_{3})&=p(x_{2})p(x_{1},x_{2},x_{3}),\label{c6} \\
		p(x_{1},x_{3})p(x_{2},x_{3})&=p(x_{3})p(x_{1},x_{2},x_{3}),\label{c7} \\
		p(x_{1},x_{2})p(x_{1},x_{4})&=p(x_{1})p(x_{1},x_{2},x_{4}),\label{c8} \\
		p(x_{1},x_{2})p(x_{2},x_{4})&=p(x_{2})p(x_{1},x_{2},x_{4}),\label{c9} \\
		p(x_{1},x_{4})p(x_{2},x_{4})&=p(x_{4})p(x_{1},x_{2},x_{4}),\label{c10} \\
		p(x_{1},x_{3})p(x_{1},x_{4})&=p(x_{1})p(x_{1},x_{3},x_{4}),\label{c11} \\
		p(x_{1},x_{3})p(x_{3},x_{4})&=p(x_{3})p(x_{1},x_{3},x_{4}),\label{c12} \\
		p(x_{1},x_{4})p(x_{3},x_{4})&=p(x_{4})p(x_{1},x_{3},x_{4}),\label{c13} \\
		p(x_{2},x_{3})p(x_{2},x_{4})&=p(x_{2})p(x_{2},x_{3},x_{4}),\label{c14} \\
		p(x_{2},x_{3})p(x_{3},x_{4})&=p(x_{3})p(x_{2},x_{3},x_{4}),\label{c15} \\
		p(x_{2},x_{4})p(x_{3},x_{4})&=p(x_{4})p(x_{2},x_{3},x_{4}).\label{c16} 
	\end{align}
	Since $p(x_1,x_2,x_3)=p(x_1,x_2,x_4)$, together with \eqref{c5} and \eqref{c8}, we have
	\begin{equation}
		p(x_1,x_3)=p(x_1,x_4).   \label{c23}
	\end{equation}
	In light of \eqref{c18} and \eqref{c19}, replacing $p(x_1,x_3)$ and $p(x_1,x_4)$ by $p(x_1)p(x_3)$ and $p(x_1)p(x_4)$  in \eqref{c23}, we obtain
	\begin{equation}
		p(x_3)=p(x_4) .\label{c24}
	\end{equation}
	Note that $X_3$ and $X_4$ are independent,  by Lemma \ref{lem1},
	$X_3$ and $X_4$ are uniformly distributed on $\mc{X}_3 $ and $\mc{X}_4$, respectively, and so $H(X_3)=H(X_4)=\log v$ where $v=|\mathcal{X}_{3}|=|\mathcal{X}_{4}|$.
	
	For the ``if'' part, it can be proved by Lemma \ref{lem}  and the fact that $a=\log{v}$ on the ray $U_{3,4}$, and
	whole ray $U^1_{1,1}$ are entropic.
\end{proof}
\begin{theorem}
	\label{rk9}
	For $F=(U_{3,4} , U^{12}_{1,2})$, $\mathbf{h}=(a,b) \in F$ is entropic if and only if $a$= $\log{v}$ for integer $v\ge 1$ . 
\end{theorem}

\begin{proof}
	If $\mathbf{ h} \in  F $ is entropic, its characterizing random vector
	$(X_i,i\in N_4)$ satisfies the following information equalities,
	\begin{align}
		H(  X_{N_4})&=H(X_{N_4-i})  ,\ i\in N_4 \nonumber \\
		H(X_{ij })&=H(X_{i})+H(X_{j}) , i<j, \{i,j\} \neq \{1,2\}   \nonumber \\
		H(X_{ ik})+H(X_{  jk})&=H(X_{ k})+H(X_{ijk } ), \{i,j\} \neq \{1,2\}.\nonumber 
	\end{align}
  For $ (x_i,i\in N_4) \in
\mathcal{X}_{N_4}$ with $p(x_{1234})>0$, above information equalities
imply that the probability mass function satisfies
	\begin{align}
		p(x_{1},x_{2},x_{3},x_{4})&=p(x_{1},x_{2},x_{3}) \label{c27} \\
		&=p(x_{1},x_{2},x_{4})  \label{c28}\\ 
		&=p(x_{1},x_{3},x_{4})  \label{c29}   \\
		&=p(x_{2},x_{3},x_{4}) , \label{c30}   \\
		p(x_1,x_3)&=p(x_1)p(x_3) ,  \label{c31}   \\
		p(x_1,x_4)&=p(x_1)p(x_4),   \label{c32}   \\
		p(x_2,x_3)&=p(x_2)p(x_3) ,  \label{c33}   \\
		p(x_2,x_4)&=p(x_2)p(x_4) ,  \label{c34}   \\
		p(x_3,x_4)&=p(x_3)p(x_4) ,  \label{c35}   \\		
		p(x_{1},x_{2})p(x_{1},x_{3})&=p(x_{1})p(x_{1},x_{2},x_{3}),\label{c36} \\
		p(x_{1},x_{2})p(x_{2},x_{3})&=p(x_{2})p(x_{1},x_{2},x_{3}),\label{c37} \\
		p(x_{1},x_{2})p(x_{1},x_{4})&=p(x_{1})p(x_{1},x_{2},x_{4}),\label{c38} \\
		p(x_{1},x_{2})p(x_{2},x_{4})&=p(x_{2})p(x_{1},x_{2},x_{4}),\label{c39} \\
		p(x_{1},x_{3})p(x_{1},x_{4})&=p(x_{1})p(x_{1},x_{3},x_{4}),\label{c40} \\
		p(x_{1},x_{3})p(x_{3},x_{4})&=p(x_{3})p(x_{1},x_{3},x_{4}),\label{c41} \\
		p(x_{1},x_{4})p(x_{3},x_{4})&=p(x_{4})p(x_{1},x_{3},x_{4}),\label{c42} \\
		p(x_{2},x_{3})p(x_{2},x_{4})&=p(x_{2})p(x_{2},x_{3},x_{4}),\label{c43} \\
		p(x_{2},x_{3})p(x_{3},x_{4})&=p(x_{3})p(x_{2},x_{3},x_{4}),\label{c44} \\
		p(x_{2},x_{4})p(x_{3},x_{4})&=p(x_{4})p(x_{2},x_{3},x_{4}).\label{c45} 
	\end{align}	
	Since $p(x_1,x_2,x_3)=p(x_1,x_2,x_4)$, together with \eqref{c36} and \eqref{c38}, we have
	\begin{equation}
		p(x_1,x_3)=p(x_1,x_4). \label{c46}
	\end{equation} 
 By \eqref{c31} and \eqref{c32}, replacing $p(x_1,x_3)$ and $p(x_1,x_4)$ by $p(x_1)p(x_3)$ and $p(x_1)p(x_4)$  in \eqref{c46},
	\begin{equation}
		p(x_3)=p(x_4).   \label{c47}
	\end{equation}
	By \eqref{c35}, $X_3$  and $X_4$ are independent. Due to Lemma \ref{lem1}, we conclude that $X_3$ and $X_4$ are uniformly distributed on $\mc{X}_3$ and  $\mc{X}_4$, respectively, and so   $H(X_3)=H(X_4)=\log v$ where $v=|\mathcal{X}_{3}|=|\mathcal{X}_{4}|$.
	
	On the other hand,  the ``if''  part of the theorem can be proved by Lemma \ref{lem}  and the fact that $a=\log{v}$ on the ray $U_{3,4}$ and
	whole ray $U^{12}_{1,2}$ are entropic.
\end{proof}
\begin{theorem}
	\label{rk10}
	For $F=(  U_{3,4}  ,U^{123}_{2,3} )$, $\mathbf{h}=(a,b) \in F
	$ is entropic if and only if  $a+b$= $\log{v}$ for integer $v\ge 1$ . 
\end{theorem}
\begin{figure}[H]
 	\centering
    \includegraphics[scale=1]{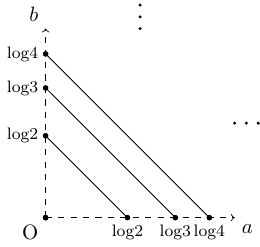}
	\caption{The face $  ( U_{3,4} , U^{123}_{2,3}  ) $ }
	\label{fig7}
\end{figure} 
 \begin{proof}
	If $\mathbf{ h} \in  F $ is entropic, its characterizing random vector
	$(X_i,i\in N_4)$ satisfies the following information equalities,
	\begin{align}
		H(  X_{N_4})&=H(X_{N_4-i})  ,\ i\in N_4 \nonumber \\
		H(X_{ij })&=H(X_{i})+H(X_{j}) ,i<j, i,j\in N_4,   \nonumber \\
		H(X_{ ik})+H(X_{  jk})&=H(X_{ k})+H(X_{ijk } ), \{i,j,k\} \neq  \{1,2,3\}.\nonumber 
	\end{align}
 For $ (x_i,i\in N_4) \in
\mathcal{X}_{N_4}$ with $p(x_{1234})>0$, above information equalities
imply that the probability mass function satisfies
	\begin{align}
		p(x_{1},x_{2},x_{3},x_{4})&=p(x_{1},x_{2},x_{3}) \label{c48} \\
		&=p(x_{1},x_{2},x_{4})  \label{c49}\\ 
		&=p(x_{1},x_{3},x_{4})  \label{50}   \\
		&=p(x_{2},x_{3},x_{4}) , \label{c51}   \\
		p(x_1.x_2)&=p(x_1)p(x_2),   \label{c52}   \\
		p(x_1.x_3)&=p(x_1)p(x_3) ,  \label{c53}   \\
		p(x_1.x_4)&=p(x_1)p(x_4) ,  \label{c54}   \\
		p(x_2.x_3)&=p(x_2)p(x_3) ,  \label{c55}   \\
		p(x_2.x_4)&=p(x_2)p(x_4) ,  \label{c56}   \\
		p(x_3.x_4)&=p(x_3)p(x_4) ,  \label{c57}   \\		
		p(x_{1},x_{2})p(x_{1},x_{4})&=p(x_{1})p(x_{1},x_{2},x_{4}),\label{c58} \\
		p(x_{1},x_{2})p(x_{2},x_{4})&=p(x_{2})p(x_{1},x_{2},x_{4}),\label{c59} \\
		p(x_{1},x_{4})p(x_{2},x_{4})&=p(x_{4})p(x_{1},x_{2},x_{4}),\label{c60} \\
		p(x_{1},x_{3})p(x_{1},x_{4})&=p(x_{1})p(x_{1},x_{3},x_{4}),\label{c61} \\
		p(x_{1},x_{3})p(x_{3},x_{4})&=p(x_{3})p(x_{1},x_{3},x_{4}),\label{c62} \\
		p(x_{1},x_{4})p(x_{3},x_{4})&=p(x_{4})p(x_{1},x_{3},x_{4}),\label{c63} \\
		p(x_{2},x_{3})p(x_{2},x_{4})&=p(x_{2})p(x_{2},x_{3},x_{4}),\label{c64} \\
		p(x_{2},x_{3})p(x_{3},x_{4})&=p(x_{3})p(x_{2},x_{3},x_{4}),\label{c65} \\
		p(x_{2},x_{4})p(x_{3},x_{4})&=p(x_{4})p(x_{2},x_{3},x_{4}).\label{c66} 
	\end{align}
	By \eqref{c52} and \eqref{c54}, replacing $p(x_1,x_2)$ and $p(x_1,x_4)$ by $p(x_1)p(x_2)$ and $p(x_1)p(x_4)$ in \eqref{c58}, we have
	\begin{equation}
		p(x_1,x_2,x_4)=p(x_1)p(x_2)p(x_4). \label{c67}
	\end{equation}
	Similarly, we obtain
	\begin{align}
		p(x_1,x_3,x_4)&=p(x_1)p(x_3)p(x_4).  \label{c68} \\
		p(x_2,x_3,x_4)&=p(x_2)p(x_3)p(x_4).      \label{c69}
	\end{align}
 Together with \eqref{c49}-\eqref{c51},
	\begin{equation}
		p(x_1)=p(x_2)=p(x_3). \label{c70}
	\end{equation}
	By  \eqref{c52}, \eqref{c53} and \eqref{c55}, $X_1$, $X_2$ and $X_3$ are pairwise independent. Then by Lemma \ref{lem1}, we conclude that $X_i$ are uniformly distributed on $\mc{X}_i$ for $i\in N_3$ and $H(X_1)=H(X_2)=H(X_3)=\log v$ where $v=|\mathcal{X}_{1}|=|\mathcal{X}_{2}|=|\mathcal{X}_{3}|$.
	
	Now we prove the ``if '' part. Let $X_2$, $X_3$ and $X_4$ be mutually independent  and distributed
	on $\mathbb{I}_v$, where $X_2$ and $X_3$ are uniform and $H(X_4)=a$ for some $a\leq \log v$. Let $X_1=X_2+X_3+X_4 \mod v$. Then for any $i\in \mc{X}_1$,
	\begin{align}
		p_{X_1}(i)=\sum_{x_2,x_3,x_4: \atop x_2+x_3+x_4 \equiv i \mod v} p_{X_2,X_3,X_4}(x_2.x_3,x_4).
	\end{align} 
	Since  $X_2$, $X_3$ and $X_4$ are mutually independent, we have
	\begin{equation}
		p_{X_1}(i)=\sum_{x_2,x_3,x_4: \atop x_2+x_3+x_4\equiv i \mod v} p_{X_2}(x_2)p_{X_3}(x_3)p_{X_4}(x_4)
	\end{equation}
	Fix $x_3$ and $x_4$,
	\begin{equation}
		p_{X_1}(i)=\sum_{j=0}^{v-1} \sum_{m=0}^{v-1} \sum_{x_2: \atop x_2+m+j\equiv i \mod v} p_{X_2}(x_2)p_{X_3}(m)p_{X_4}(j).
	\end{equation}
	Note that 
	\begin{equation}
		|\{x_2 \in\mc{X}_2: x_2+m+j\equiv i \mod v \}|=1,
	\end{equation} 
	we have
	\begin{align}
		p_{X_1}(i)&=\sum_{j=0}^{v-1} \sum_{m=0}^{v-1} \dfrac{1}{v}p_{X_3}(m)p_{X_4}(j) \\
		&=\dfrac{1}{v} \sum_{j=0}^{v-1}p_{X_4}(j)\sum_{m=0}^{v-1}p_{X_3}(m) \\
		&=\dfrac{1}{v},   \label{c71}
	\end{align}
	which implies $X_1$ is uniformly distributed on $\mc{X}_1$ and $H(X_1 )=\log v$. It can be checked that the entropy function of such
	constructed $(X_i, i\in N_4)$ is $(a,b)$.
\end{proof}

\begin{theorem}
	\label{rk11}
	For $F=(  U_{3,4}  ,U^{123}_{1,3} )$, $\mathbf{h}=(a,b) \in F
	$ is entropic if and only if  $a$= $\log{v}$ for integer $v\ge 1$ . 
\end{theorem}

\begin{proof}
	If $\mathbf{ h} \in  F $ is entropic, its characterizing random vector
	$(X_i,i\in N_4)$ satisfies the following information equalities,
		\begin{align}
		H(  X_{N_4})&=H(X_{N_4-i})  ,\ i\in N_4, \nonumber \\
		H(X_{i4 })&=H(X_{i})+H(X_{4}) , i\in N_3,    \nonumber \\
		H(X_{ ik})+H(X_{  jk})&=H(X_{ k})+H(X_{ijk } ),  k\neq 4,\nonumber\\&\qquad\qquad\text{ distinct }i,j,k\in N_4.  \nonumber
	\end{align}
  For $ (x_i,i\in N_4) \in
\mathcal{X}_{N_4}$ with $p(x_{1234})>0$, above information equalities
imply that the probability mass function satisfies
	\begin{align}
		p(x_{1},x_{2},x_{3},x_{4})&=p(x_{1},x_{2},x_{3}) \label{c109} \\
		&=p(x_{1},x_{2},x_{4})  \label{c110}\\ 
		&=p(x_{1},x_{3},x_{4})  \label{c111}   \\
		&=p(x_{2},x_{3},x_{4}) , \label{c112}   \\
		p(x_1,x_4)&=p(x_1)p(x_4),   \label{c113}   \\
		p(x_2,x_4)&=p(x_2)p(x_4) ,  \label{c114}   \\
		p(x_3,x_4)&=p(x_3)p(x_4) ,  \label{c115}   \\
		p(x_{1},x_{2})p(x_{1},x_{3})&=p(x_{1})p(x_{1},x_{2},x_{3}),\label{c116} \\
		p(x_{1},x_{2})p(x_{2},x_{3})&=p(x_{2})p(x_{1},x_{2},x_{3}),\label{c117} \\
		p(x_{1},x_{3})p(x_{2},x_{3})&=p(x_{3})p(x_{1},x_{2},x_{3}),\label{c118} \\
		p(x_{1},x_{2})p(x_{1},x_{4})&=p(x_{1})p(x_{1},x_{2},x_{4}),\label{c119} \\
		p(x_{1},x_{2})p(x_{2},x_{4})&=p(x_{2})p(x_{1},x_{2},x_{4}),\label{c120} \\
		p(x_{1},x_{3})p(x_{1},x_{4})&=p(x_{1})p(x_{1},x_{3},x_{4}),\label{c121} \\
		p(x_{1},x_{3})p(x_{3},x_{4})&=p(x_{3})p(x_{1},x_{3},x_{4}),\label{c122} \\
		p(x_{2},x_{3})p(x_{2},x_{4})&=p(x_{2})p(x_{2},x_{3},x_{4}),\label{c123} \\
		p(x_{2},x_{3})p(x_{3},x_{4})&=p(x_{3})p(x_{2},x_{3},x_{4}).\label{c124} 
	\end{align}
	Since $p(x_1,x_2,x_3)=p(x_1,x_2,x_4)=p(x_1,x_3,x_4)$, together with \eqref{c116}, \eqref{c119} and \eqref{c121}, we obtain
	\begin{align}
		p(x_1,x_2)=p(x_1,x_3)=p(x_1,x_4).\label{c125} 
	\end{align}
	By the same argument, we have
	\begin{align}
		p(x_1,x_2)&=p(x_2,x_3)=p(x_2,x_4),\label{c126} \\
		p(x_1,x_3)&=p(x_2,x_3)=p(x_3,x_4).\label{c127} 
	\end{align}
	In light of \eqref{c125}-\eqref{c127},
	\begin{align}
		& \quad \ p(x_1,x_2)=p(x_1,x_3)=p(x_1,x_4) \label{c128} \\
		&=p(x_2,x_3)=p(x_2,x_4)=p(x_3,x_4) ,   \label{c129} 
	\end{align}
	which implies the left hand side of \eqref{c116}-\eqref{c118} are equal. Then  we conclude that $p(x_1)=p(x_2)=p(x_3)$. By \eqref{c128}, replacing $p(x_1,x_2)$  and  $p(x_1,x_3)$ by $p(x_1,x_4)$  in \eqref{c116}, we obtain
	\begin{equation}
		p(x_1,x_4)p(x_1,x_4)=p(x_1)p(x_{1},x_{2},x_{3}).   \label{c130} 
	\end{equation}
	By \eqref{c113} and \eqref{c109}, replacing $p(x_1,x_4)$ by $p(x_1,x_4)$, $p(x_1,x_2,x_3)$ by $p(x_1,x_2,x_3,x_4)$, we have
	\begin{align}
		p(x_{1},x_{2},x_{3},x_4)=p(x_1)p^2(x_4). \label{c131} 
	\end{align} 
	By \eqref{c113}, $X_1$ and $X_4$ are indenpendent. For any $x_4' \in  \mc{X}_4$ with $x_4'\neq x_4 $,
	\begin{equation}
		p(x_1,x_4')=p(x_1)p(x_4')>0. \label{c132}
	\end{equation}
	As
	\begin{equation}
		p(x_1,x_4')=\sum_{x_2,x_3} p(x_1,x_2,x_3,x_4'),
	\end{equation}
	there exists $ x_2'\in  \mc{X}_2$, $ x_3'\in \mc{X}_3$ such that $p(x_1,x_2',x_3',x_4')>0$.
	By the same argument, we have 
	\begin{equation}
		p(x_1,x_2',x_3',x_4')=p(x_1)p^2(x_4'). \label{c133}
	\end{equation}
	Consider the tripartite graph $G=(V,E)$ with $V=\mc{X}_1\cup \mc{X}_2\cup \mc{X}_3$ and   $(x_i,x_j)\in E$  if and only if $p(x_{i},x_{j})>0$, ${i,j}\in N_3 ,i\neq j$. Then, by Lemma \ref{lem4}, the probability mass of the triangles are equal in each connected component, which implies that $p(x_1,x_2,x_3,x_4)=p(x_1,x_2',x_3',x_4')$. Equating  \eqref{c131} and \eqref{c133}, 
	we have 
	\begin{align}
		p(x_4)=p(x_4').
	\end{align}
	It holds for any $x_4'\in \mc{X}_4$, which implies that $X_4$ is uniformly distributed on $\mc{X}_4$ , and so $H(X_4)=\log v$ where $v=|\mc{X}_4|$.
	
	On the other hand,  the ``if'' part of the theorem can be proved by Lemma \ref{lem}  and the fact that $a=\log{v}$ on the ray $U_{3,4}$ and
	whole ray $U^{123}_{1,3}$ are entropic.
\end{proof}

\begin{theorem}
	\label{rk12}
	For $F=(  U_{3,4}  ,U_{1,4} )$, $\mathbf{h}=(a,b) \in F
	$ is entropic if and only if  $a\geq 0,b>0$ or $(a,b)=(\log{v},0)$
	for some positive integer $v$. 
\end{theorem}

 \begin{figure}[H]
 	\centering
 	 \includegraphics[trim=0 0 0 10,scale=1.05,clip]{fig_rk7.pdf}
 	 \caption{The face $  ( U_{3,4} , U_{1,4}  ) $ }
 \end{figure}
 
\begin{proof}
	If $\mathbf{ h} \in  F $ is entropic, its characterizing random vector
	$(X_i,i\in N_4)$ satisfies the following information equalities,
	\begin{align}
		H(  X_{N_4})&=H(X_{N_4-i})  ,\ i\in N_4, \nonumber \\
		H(X_{ ik})+H(X_{  jk})&=H(X_{ k})+H(X_{ijk } ), \text{ distinct }i,j,k\in N_4.\nonumber 
	\end{align}
  For $ (x_i,i\in N_4) \in
\mathcal{X}_{N_4}$ with $p(x_{1234})>0$, above information equalities
imply that the probability mass function satisfies
	\begin{align}
		p(x_{1},x_{2},x_{3},x_{4})&=p(x_{1},x_{2},x_{3}) \label{c72} \\
		&=p(x_{1},x_{2},x_{4})  \label{c73}\\ 
		&=p(x_{1},x_{3},x_{4})  \label{c74}   \\
		&=p(x_{2},x_{3},x_{4}) , \label{c75}   \\
		p(x_{1},x_{2})p(x_{1},x_{3})&=p(x_{1})p(x_{1},x_{2},x_{3}),\label{c76} \\
		p(x_{1},x_{2})p(x_{2},x_{3})&=p(x_{2})p(x_{1},x_{2},x_{3}),\label{c77} \\
		p(x_{1},x_{3})p(x_{2},x_{3})&=p(x_{3})p(x_{1},x_{2},x_{3}),\label{c78} \\
		p(x_{1},x_{2})p(x_{1},x_{4})&=p(x_{1})p(x_{1},x_{2},x_{4}),\label{c79} \\
		p(x_{1},x_{2})p(x_{2},x_{4})&=p(x_{2})p(x_{1},x_{2},x_{4}),\label{c80} \\
		p(x_{1},x_{4})p(x_{2},x_{4})&=p(x_{4})p(x_{1},x_{2},x_{4}),\label{c81} \\
		p(x_{1},x_{3})p(x_{1},x_{4})&=p(x_{1})p(x_{1},x_{3},x_{4}),\label{c82} \\
		p(x_{1},x_{3})p(x_{3},x_{4})&=p(x_{3})p(x_{1},x_{3},x_{4}),\label{c83} \\
		p(x_{1},x_{4})p(x_{3},x_{4})&=p(x_{4})p(x_{1},x_{3},x_{4}),\label{c84} \\
		p(x_{2},x_{3})p(x_{2},x_{4})&=p(x_{2})p(x_{2},x_{3},x_{4}),\label{c85} \\
		p(x_{2},x_{3})p(x_{3},x_{4})&=p(x_{3})p(x_{2},x_{3},x_{4}),\label{c86} \\
		p(x_{2},x_{4})p(x_{3},x_{4})&=p(x_{4})p(x_{2},x_{3},x_{4}).\label{c87} 
	\end{align}
	Since $p(x_1,x_2,x_3)=p(x_1,x_2,x_4)=p(x_1,x_3,x_4)$, by \eqref{c76},\eqref{c79} and \eqref{c82}, we obtain 
	\begin{equation}
		p(x_1,x_2)p(x_1,x_3)=p(x_1,x_2)p(x_1,x_4)=p(x_1,x_3)p(x_1,x_4),
	\end{equation}
which implies
	\begin{equation}
		p(x_1,x_2)=p(x_1,x_3)=p(x_1,x_4). \label{c900}
	\end{equation}
	By the same argument, we have
	\begin{align}
		p(x_1,x_2)&=p(x_2,x_3)=p(x_2,x_4).  \label{c901}\\
		p(x_1,x_3)&=p(x_2,x_3)=p(x_3,x_4).  \label{c902}
	\end{align}
	Equating \eqref{c900}-\eqref{c902},
	\begin{align}
		& \quad \ p(x_1,x_2)=p(x_1,x_3)=p(x_1,x_4)   \label{c88}\\ 
		&=p(x_2,x_3)=p(x_2,x_4)=p(x_3,x_4), \label{c89}
	\end{align}
	which implies the left  side of \eqref{c76}-\eqref{c87} are equal. Then with \eqref{c72}-\eqref{c75}, we conclude that
	\begin{equation}
		p(x_1)=p(x_2)=p(x_3)=p(x_4). \label{c90}
	\end{equation}
	
	Consider the tripartite graph $G=(V,E)$ with $V=\mc{X}_1\cup \mc{X}_2\cup \mc{X}_3$ and   $(x_i,x_j)\in E$  if and only if $p(x_{i},x_{j})>0$, ${i,j}\in N_3 ,i\neq j$. Assume $G$ has $t$ connected components.  \rv{We denote the number of the vertices of $\mc{X}_{i}$  in the connected component $C_{j}$ by
	$n_{i}^{(j)}, i=1,2,3, j=1,2,\cdots,t$}  and  the  probability mass of $ C_{j}$ by $p_{j}$. By  Lemma \ref{lem4},  we obtain $n_1^{(j)}=n_2^{(j)}=n_3^{(j)}$ and so it can be simplified to  $n^{(j)}$. The probability mass of the vertices, the edges  and the triangles in $C_j$ are
	\begin{align}
		&p(x_1)=p(x_2)=p(x_3)=\dfrac{p_j}{n^{(j)}} ,  \label{c91} \\  
		&p(x_1,x_2)=p(x_1,x_3)=p(x_2,x_3)=\dfrac{p_j}{(n^{(j)})^2},  \label{c92}   \\ 
		&p(x_1,x_2,x_3)=\dfrac{p_j}{(n^{(j)})^3},   \label{c93}
	\end{align}
respectively. Then
\begin{align}
	H(X_1)&=H(X_2)=H(X_3) \label{c107}  \\
	&=-\sum_{j=1}^t  p_j \log \dfrac{p_j}{n^{(j)}}  \label{c94} \\
	&=H(p_1,...,p_t)+\sum_{j=1}^t p_j\log n^{(j)}  , \label{c95}   \\
	H(X_1,X_2)&=H(X_1,X_3)=H(X_2,X_3)   \\
	&=-\sum_{j=1}^t  p_j \log \dfrac{p_j}{(n^{(j)})^2}   \label{c96} \\
	&=H(p_1,...,p_t)+2\sum_{j=1}^t p_j\log n^{(j)} , \label{c97}
\end{align}
\begin{align}
	H(X_1,X_2,X_3)=-\sum_{j=1}^t  p_j \log \dfrac{p_j}{(n^{(j)})^3}  \label{c98}  \\
	=H(p_1,...,p_t)+3\sum_{j=1}^t p_j\log n^{(j)} . \label{c99}
\end{align}
	As $\h\in F$ and $(X_i, i\in N_4)$ is
	its characterizing random vector, we have 
	\begin{align}
		H(X_1)&=H(X_2)=	H(X_3) \label{c100} \\
		&=a+b, \label{c101} \\
		H(X_1,X_2)&=H(X_1,X_3)=H(X_2,X_3)  \label{c102} \\
		&=2a+b ,\label{c103} \\
		H(X_1,X_2,X_3)&=3a+b .\label{c104} 
	\end{align}
	By \eqref{c107}-\eqref{c104}, we have
	\begin{align}
		a&=\sum_{j=1}^t p_j\log n^{(j)}, \label{c105}  \\
		b&=H(p_1,...,p_t).   \label{c106}
	\end{align}
	
	By Lemma \ref{lem} and the fact that $a=\log v$ on the ray $U_{3,4}$, and the whole ray $U_{1,4}$ are entropic, all $\h=(a,b)\in F$ are entropic  when $a=\log v $ for positive integer $v$, and $b\geq 0$.
    In the following, we only need to show that for all $\h= (a,b)\in F$ are
    entropic, when $a>0,b>0$ with $a\neq \log v$  for some integer $v\ge 1$. 
	

Consider $a\neq\log{v},b>0$. Let $c=\lfloor
2^a\rfloor$. Let $n^{(j)}=c$ for all $j\in N_{t-1}$, $n^{(t)}=c+l$ for
some sufficiently large $l$ and  $0<p_i< 1$, $i=1,2,\dots$, \rv{$t$} be positive values satisfying $H(p_1,...,p_t)=b$ such that 
	\begin{align}
		p_t=\dfrac{a-\log{\lfloor 2^a\rfloor}}{\log{(\lfloor 2^a\rfloor+l)}-\log{\lfloor 2^a\rfloor}}. \label{c903}
	\end{align}
With such construction, we can check  \eqref{c105} is stasified with
\begin{align}
	\sum_{j=1}^t p_j\log n^{(j)}&=\sum_{j=1}^{t-1} p_j\log c +  p_t\log{(c+l)}\\
	&=(1-p_t)\log c+ p_t\log{(c+l)}\\&=a.
\end{align}
	Let 
$X^{(j)}_{1}$ ,$X^{(j)}_{2}$ and $X^{(j)}_{3}$ be mutually independent and uniformly distributed on 
$\mathbb{I}_c$ , and $X^
{(j)}_{4}=X^ {(j)}_{1}+X^ {(j)}_{2}+X^ {(j)}_{3} \mod c$  for $j\in N_{t-1}$.
Let $X^{(t)}_{1}$, $X^{(t)}_{2}$ and $X^{(t)}_{3}$ be mutually independent
and uniformly distributed on $\mathbb{I}_{c+l}$ and 
$X^{(t)}_{4}=X^{(t)}_{1}+X^{(t)}_{2}+X^{(t)}_{3} \mod (c+l)$. Assume $X^{(j)}_{i}$ is distributed on $\mathcal{X}^{(j)}_{i}$ for $i\in N_4$ and $j\in N_t$. Let $X_i,i\in N_4$ be distributed on $\mathbb{I}_{tc+l}$ such that 
\begin{align}
	&p_{X_{N_4}}(x_i+(j-1)c, i\in N_4) \nonumber\\=&p_{j}p_{X^{(j)}_1,X^{(j)}_2,X^{(j)}_3,X^{(j)}_4}(x_1,x_2,x_3,x_4)
\end{align}
for $ j\in N_{t}$ and $(x_1,x_2,x_3,x_4)\in\mathcal{X}^{(j)}_{N_4}$.
 It can be checked that the entropy function of such
	constructed $(X_i,i\in N_4)$ is $(a,b)$.
	Hence, all $\h=(a,b)\in F$ are entropic when
	$a>0,b>0$. 
	\rem{Note that $a=\log k$ on the ray $U_{3,4}$ and the whole ray $U_{1,4}$ are entropic.}	 The proof is accomplished.
\end{proof}
Theorem \ref{rk12} can be considered as a corollary of \cite[Lemma 4]{matuvs2007},
which is restated as \cite[Lemma 6]{csirmaz2025exp}.

\subsection{Entropy functions on the faces involving $\hat{U}_{3,5}^{1}$}

\label{G}
\begin{theorem}
	\label{rk22}
	For   $F= (\hat{U}_{3,5}^{1},U_{1,3}^{234})$ , $\h=(a,b)\in F$ is entropic if and only if $a=\log v $ for integer $v\ge 1$.
\end{theorem}
 
\begin{proof}
	If $\mathbf{ h} \in  F$ is entropic, its characterizing random vector
	$(X_i,i\in N_4)$ satisfies the following information equalities,
	\begin{align}
	H(  X_{N_4})&=H(X_{N_4-i})  ,\ i\in N_4, \nonumber \\
	H(X_{1i })&=H(X_{1})+H(X_{i}), i\in\{2,3,4\}, \nonumber \\
	H(X_{  ik })+H(X_{jk })&=H(X_{ k})+H(X_{  ijk} ),\{i,j,k\} =\{2,3,4\},\nonumber \\
	H(X_{i\cup K})+H(X_{  j\cup K})&=H(X_{ K})+H(X_{  ij
		\cup K }), | K|=2,\nonumber\\ &\qquad\qquad\qquad\qquad\{1\} \subseteq K\subseteq N_4.   \nonumber
\end{align}
  For $ (x_i,i\in N_4) \in
\mathcal{X}_{N_4}$ with $p(x_{1234})>0$, above information equalities
imply that the probability mass function satisfies
	\begin{align}
		p(x_{1},x_{2},x_{3},x_{4})&=p(x_{1},x_{2},x_{3}) \label{c293}\\
		&=p(x_{1},x_{2},x_{4}) \label{c294} \\ 
		&=p(x_{1},x_{3},x_{4}) \label{c295}\\
		&= p(x_{2},x_{3},x_{4}), \label{c296}\\
		p(x_1,x_2)&=p(x_1)p(x_2), \label{c297}\\
		p(x_1,x_3)&=p(x_1)p(x_3), \label{c298}\\
		p(x_1,x_4)&=p(x_1)p(x_4), \label{c299}\\
		p(x_2,x_3)p(x_2,x_4)&=p(x_2)p(x_2,x_3,x_4), \label{c400}\\
		p(x_2,x_3)p(x_3,x_4)&=p(x_3)p(x_2,x_3,x_4), \label{c401}\\
		p(x_2,x_4)p(x_3,x_4)&=p(x_4)p(x_2,x_3,x_4), \label{c402}\\	p(x_{1},x_{2},x_{3})p(x_{1},x_{2},x_{4})&=p(x_{1},x_{2})p(x_{1234}),\label{c403}\\
		p(x_{1},x_{2},x_{3})p(x_{1},x_{3},x_{4})&=p(x_{1},x_{3})p(x_{1234}),\label{c404}\\
		p(x_{1},x_{2},x_{4})p(x_{1},x_{3},x_{4})&=p(x_{1},x_{4})p(x_{1234}).\label{c405}
	\end{align}
		Consider the tripartite graph $G=(V,E)$ with $V=\mc{X}_2\cup \mc{X}_3\cup \mc{X}_4$ and   $(x_i,x_j)\in E$  if and only if $p(x_{i},x_{j})>0$, ${i,j}\in \{2,3,4\} ,i\neq j$.  By \eqref{c400}-\eqref{c402} and Lemma \ref{lem4}, each connected component is a complete tripartite graph.
	By \eqref{c293}, canceling $p(x_1,x_2,x_3)$ and $p(x_1,x_2,x_3,x_4)$ on either side of \eqref{c403}, we have
	\begin{align}
		p(x_1,x_2,x_4)=p(x_1,x_2)  \label{c406}.
	\end{align}
	Similarly, by \eqref{c293}, \eqref{c294}, \eqref{c404} and \eqref{c405}, we obtain
	\begin{align}
		p(x_1,x_3,x_4)=p(x_1,x_3) , \label{c407} \\
		p(x_1,x_3,x_4)=p(x_1,x_4) .\label{c408}
	\end{align} 
	In light of \eqref{c293}-\eqref{c299},
		\begin{align}
		&p(x_{1234})=p(x_1.x_2)=p(x_1,x_3)=p(x_1,x_4)  \label{c409}\\
		=&p(x_1)p(x_2)=p(x_1)p(x_3)=p(x_1)p(x_4), \label{c410}
	\end{align}
	which implies $p(x_2)=p(x_3)=p(x_4)$. By Lemma \ref{lem4}, the number of vertices in $\mc{X}_i$, $i=1,2,3$ are the same  and the the probability mass of all of the vertices, the edges and the triangles  are equal, respectively, in each connected component. 
	 Assume $G$ has $t$ connected components, $n^{(j)}$ and $p_j$  denote the number of the vertices and the probability mass of the connected component $C_j$, respectively. Then, the probability mass of the vertices, the edges and the triangles are $\frac{p_j}{n^{(j)}},\frac{p_j}{{(n^{(j)})}^2}$ and $ \frac{p_j}{{(n^{(j)})}^3}$, respectively. Then,
	\begin{align}   
	H(X_2)&=H(X_3)=H(X_4)  \label{c411}\\
	&	=-\sum_{i=1}^t n^{(j)} \dfrac{p_j}{n^{(j)}} \log{\dfrac{p_j}{n^{(j)}}} \label{c412}\\
	&   =H(p_1,p_2,...,p_t)+ \sum_{i=1}^t p_j \log{n^{(j)}},\label{c413}\\
	H(X_2,X_3)&=H(X_2,X_4)=H(X_3,X_4)\label{c414}\\
	&	=-\sum_{i=1}^t {(n^{(j)})}^2\dfrac{p_j}{{(n^{(j)})}^2} \log{\dfrac{p_j}{{(n^{(j)})}^2}} \label{c415}\\
	&   =H(p_1,p_2,...,p_t)+ 2\sum_{i=1}^t p_j \log{n^{(j)}},\label{c416}
\end{align}
\begin{align}
	H(X_2&,X_3,X_4)=-\sum_{i=1}^t {(n^{(j)})}^3\dfrac{p_j}{{(n^{(j)})}^3} \log{\dfrac{p_j}{{(n^{(j)})}^3}} \nonumber\\  &=H(p_1,p_2,...,p_t)+ 3\sum_{i=1}^t p_j \log{n^{(j)}}.\label{c418}
\end{align}
	As $\h\in F$, $(X_i, i\in N_4)$ is
	its characterizing random vector, we have 
	\begin{align}
		H(X_2)&=H(X_3)=	H(X_4) \label{c419} \\
		&=a+b, \label{c420} \\
		H(X_2,X_3)&=H(X_2,X_4)=H(X_3,X_4)  \label{c421} \\
		&=2a+b ,\label{c422} \\
		H(X_2,X_3,X_4)&=3a+b .\label{c423} 
	\end{align}
	By\eqref{c411}-\eqref{c423}, we have
	\begin{align}
		a&=\sum_{j=1}^t p_j\log n^{(j)}, \label{c424}  \\
		b&=H(p_1,...,p_t).  \label{c425}
	\end{align}
By \eqref{c296}, $p(x_2,x_3,x_4)=p(x_1,x_2,x_3,x_4)$, and so each triangle $(x_2,x_3,x_4)$ can be colored by a unique $x_1\in\mc{X}_1$.  For any $x_2\in\mc{X}_2$ in $C_j$ , there are  ${(n^{(j)})}^2$ triangles containing the vertex $x_2$
based on the fact that each  connected component is a complete tripartite graph. By \eqref{c409}, $p(x_1,x_2,x_3,x_4)=p(x_1,x_2)$,  which implies that these ${(n^{(j)})}^2$ triangles are colored differently. Note that $X_1$ and $X_2$ are independent by \eqref{c297}, we have
	\begin{align}
		| \mc{X}_1|={(n^{(j)})}^2. \label{c426}
	\end{align}
 Note that it holds for all connected components of $G$, we obtain
	\begin{align}
		| \mc{X}_1|={(n^{(1)})}^2={(n^{(2)})}^2=...={(n^{(t)})}^2,\label{c427}
	\end{align} 
which implies that 
\begin{equation}
	n^{(1)}=n^{(2)}=\dots=n^{(t)}
\end{equation}
and so they can be simplified to $v$. Then by \eqref{c424}, 
	\begin{align}
		a&=\sum_{j=1}^t p_j\log v =\log v\sum_{j=1}^t p_j=\log v.\label{c428}
	\end{align}
	Hence $a$ can just take  the value of $\log v$, where $v$ is a positive integer.
	
%
%
	As for  ``if'' part of the theorem, it can be proved by Lemma \ref{lem}  and the fact that $a=\log{v}$ on the ray $\hat{U}^1_{3,5}$ and
	whole ray $U_{1,3}^{234}$ are entropic.
\end{proof}

\begin{theorem}
	\label{rk23}
	For $F= (\hat{U}_{3,5}^{1},U_{2,3}^{123})$, $\mathbf{h} =(a,b)\in
	F $ is entropic if and only if $a= \log{v_{1}}, b= \log{v_{2}}  $
	for positive integer $v_{1},v_{2}$.
\end{theorem}
 
\begin{proof}
	If $\mathbf{ h} \in  F$ is entropic, its characterizing random vector
	$(X_i,i\in N_4)$ satisfies the following information equalities,
	\begin{align}
	H(  X_{N_4})&=H(X_{N_4-i})  ,\ i\in N_4, \nonumber \\
	H(X_{ij })&=H(X_{i})+H(X_{j}), i<j,i,j\in N_4, \nonumber \\
	H(X_{  ik })+H(X_{jk })&=H(X_{ k})+H(X_{  i,j,k} ),\nonumber\\  &\qquad\qquad\qquad \quad\{i,j,k\} =\{2,3,4\},\nonumber \\
	H(X_{i\cup K})+H(X_{  j\cup K})&=H(X_{ K})+H(X_{  ij
		\cup K }), \nonumber\\&\qquad \qquad \qquad\quad K=\{1,2\},\{1,3\}.   \nonumber
\end{align}
 For $ (x_i,i\in N_4) \in
\mathcal{X}_{N_4}$ with $p(x_{1234})>0$, above information equalities
imply that the probability mass function satisfies
	\begin{align}
		p(x_{1},x_{2},x_{3},x_{4})&=p(x_{1},x_{2},x_{3}) \label{c429}\\
		&=p(x_{1},x_{2},x_{4}) \label{c430} \\ 
		&=p(x_{1},x_{3},x_{4}) \label{c431}\\
		&= p(x_{2},x_{3},x_{4}), \label{c432}\\
		p(x_1,x_2)&=p(x_1)p(x_2), \label{c433}\\
		p(x_1,x_3)&=p(x_1)p(x_3), \label{c434}\\
		p(x_1,x_4)&=p(x_1)p(x_4), \label{c435}\\
		p(x_2,x_3)&=p(x_2)p(x_3), \label{c436}\\
		p(x_2,x_4)&=p(x_2)p(x_4), \label{c437}\\
		p(x_3,x_4)&=p(x_3)p(x_4), \label{c438}\\
		p(x_2,x_3)p(x_2,x_4)&=p(x_2)p(x_2,x_3,x_4), \label{c439}\\
		p(x_2,x_3)p(x_3,x_4)&=p(x_3)p(x_2,x_3,x_4), \label{c440}\\
		p(x_2,x_4)p(x_3,x_4)&=p(x_4)p(x_2,x_3,x_4), \label{c441}\\	p(x_{1},x_{2},x_{3})p(x_{1},x_{2},x_{4})&=p(x_{1},x_{2})p(x_{1234}),\label{c442}\\
		p(x_{1},x_{2},x_{3})p(x_{1},x_{3},x_{4})&=p(x_{1},x_{3})p(x_{1234}).\label{c443}
	\end{align}
	By \eqref{c429}-\eqref{c432}, \eqref{c442} and \eqref{c443}, we have
	\begin{align}
		p(x_1,x_2,x_3,x_4)=p(x_1,x_2)=p(x_1,x_3).  \label{c444}
	\end{align}
	As $X_1$ is independent of  $X_2$ and $X_3$,
	\begin{align}
		p(x_1,x_2,x_3,x_4)=p(x_1)p(x_2)=p(x_1)p(x_3),  \label{c445}
	\end{align}
	which implies $p(x_2)=p(x_3)$. Since $X_2$ and $X_3$ are independent by \eqref{c436}, with Lemma \ref{lem1}, $X_2$ and $X_3$ are uniformly distributed on $	\mathcal{X}_{2} $ and 	$\mathcal{X}_{3}$ and 
	\begin{align}
		H(X_2)=H(X_3)=\log v, \label{c1011}
	\end{align}
 where $v=|\mathcal{X}_{2}|=|\mathcal{X}_{3}|$. Note that $X_2$ is independent of  $X_3$ and $X_4$, replacing  $p(x_2,x_3)$ and $p(x_2,x_4)$ by $p(x_2)p(x_3)$ and $p(x_2)p(x_4)$ in \eqref{c439}, respectively, we have
	\begin{align}
		p(x_2,x_3,x_4)=p(x_2)p(x_3)p(x_4).  \label{c446}
	\end{align}
	Since $p(x_1,x_2,x_3,x_4)=p(x_2,x_3,x_4)$ by \eqref{c432}, together with \eqref{c445} and \eqref{c446}, 
	\begin{align}
		p(x_1)=p(x_3)p(x_4)=p(x_2)p(x_4).  \label{c447}
	\end{align}
	Note that $p(x_2)=p(x_3)=\frac{1}{v}$, we obtain
	\begin{equation}
		p(x_1)=\dfrac{1}{v} p(x_4) . \label{c448}
	\end{equation}
	Since $X_1$ and $X_4$ are independent, for any $x_4'\in \mc{X}_4$ with $x_4'\neq x_4$, 
	\begin{equation}
		p(x_1,x_4')=p(x_1)p(x_4'). \label{c449}
	\end{equation}
	As $p(x_1,x_4')=\sum_{x_2,x_3}p(x_1,x_2,x_3,x_4')$, there exists $x_2'\in\mc{X}_2, x_3'\in \mc{X}_3$ such that $p(x_1,x_2',x_3',x_4')>0$. By the same argument, we have
	\begin{align}
		p(x_1)=\dfrac{1}{v} p(x_4') . \label{c450}
	\end{align}
	Equating \eqref{c448} and \eqref{c450}, 
	\begin{equation}
		p(x_4)=p(x_4'),
	\end{equation}
	which implies that $X_4$ is  uniformly distributed on $	\mathcal{X}_{4}$. Similarly, $X_1$ is  uniformly distributed on $	\mathcal{X}_{1}$. Assume $| X_4 | =v_1$, we conclude that
	\begin{align}
		 H(X_1)=\log vv_1,H(X_4)=\log v_1. \label{c1010}
	\end{align} 
	For any $x_1 \in \mathcal{X}_{1}, x_4 \in \mathcal{X}_{4} $, due to the independence of $X_1$ and $X_4$,
	\begin{align}
		p(x_1,x_4)=p(x_1)p(x_4)=\dfrac{1}{vv_1}\times\dfrac{1}{v_1}=\dfrac{1}{vv^2_1}.   \label{c451}
	\end{align}
	As
	\begin{align}
		p(x_1,x_4)&=\sum_{x_2,x_3}p(x_{1234})=\sum_{x_2,x_3}p(x_2,x_3,x_4)   \label{c452} \\
		&=\sum_{x_2,x_3}p(x_2)p(x_3)p(x_4)=\dfrac{v_2}{v^2v_1},  \label{c453}
	\end{align}
	where $v_2=| \{(x_2,x_3)|\quad p(x_1,x_2,x_3,x_4)>0\} |$, we conclude that $v=v_1v_2$.
	As $\h\in F$ and  $(X_i, i\in N_4)$ is its characterizing random vector, we have 
	\begin{align}
		H(X_1)&=2a+b,   \label{c454}\\
		H(X_2)&=H(X_3)=a+b, \label{c455} \\
		H(X_4)&=a.  \label{c456}
	\end{align}
Together with \eqref{c1011} and \eqref{c1010}, we obtain
	\begin{align}
		a=\log v_1,  b=\log \frac{v}{v_1}=\log v_2, \label{c460}
	\end{align}
	which forms an outer bound on the entropy region on $F$, that is, $a$  and $b$ must take the value of $\log v_1$ and $\log v_2$, respectively, where   $v_1$ and $v_2$ are positive integers.
	
	 The \rv{``if''} part of the theorem can be immediately proved using Lemma \ref{lem}, given the fact that $a = \log{v_1}$ in $\hat{U}^1_{3,5}$, and $b = \log{v_2}$ in $U_{2,3}^{123}$, both of which are entropic. 
\end{proof}

\begin{theorem}
	\label{rk24}
	For $F= (\hat{U}_{3,5}^{1},\mc{W}_{2}^{23})$, $\mathbf{h} =(a,b)\in
	F $ is entropic if and only if  $a+b$= $\log{v}$ for integer $v\ge 1$.
\end{theorem}
  \begin{proof}
	If $\mathbf{ h} \in  F$ is entropic, its characterizing random vector
	$(X_i,i\in N_4)$ satisfies the following information equalities,
	\begin{align}
		H(  X_{N_4})&=H(X_{N_4-i})  ,\ i\in N_4, \nonumber \\
		H(X_{ij })&=H(X_{i})+H(X_{j}),  \{i,j\}\neq \{2,3\} \nonumber \\
		H(X_{  23 })+H(X_{k4})&=H(X_{ k})+H(X_{ 2,3,4} ), k\in\{2,3\} \nonumber \\
		H(X_{i\cup K})+H(X_{  j\cup K})&=H(X_{ K})+H(X_{  ij
			\cup K }),|K|=2, \{1\}\subseteq K.  \nonumber
	\end{align}
  For $ (x_i,i\in N_4) \in
\mathcal{X}_{N_4}$ with $p(x_{1234})>0$, above information equalities
imply that the probability mass function satisfies
	\begin{align}
		p(x_{1},x_{2},x_{3},x_{4})&=p(x_{1},x_{2},x_{3}) \label{c461}\\
		&=p(x_{1},x_{2},x_{4}) \label{c462} \\ 
		&=p(x_{1},x_{3},x_{4}) \label{c463}\\
		&= p(x_{2},x_{3},x_{4}), \label{c464}\\
		p(x_1,x_2)&=p(x_1)p(x_2), \label{c465}\\
		p(x_1,x_3)&=p(x_1)p(x_3), \label{c466}\\
		p(x_1,x_4)&=p(x_1)p(x_4), \label{c467}\\
		p(x_2,x_4)&=p(x_2)p(x_4), \label{c468}\\
		p(x_3,x_4)&=p(x_3)p(x_4), \label{c469}\\
		p(x_2,x_3)p(x_2,x_4)&=p(x_2)p(x_2,x_3,x_4), \label{c470}\\
		p(x_2,x_3)p(x_3,x_4)&=p(x_3)p(x_2,x_3,x_4), \label{c471}\\
		p(x_{1},x_{2},x_{3})p(x_{1},x_{2},x_{4})&=p(x_{1},x_{2})p(x_{1234}),\label{c472}\\
		p(x_{1},x_{2},x_{3})p(x_{1},x_{3},x_{4})&=p(x_{1},x_{3})p(x_{1234}),\label{c473}\\
		p(x_{1},x_{2},x_{4})p(x_{1},x_{3},x_{4})&=p(x_{1},x_{4})p(x_{1234}).\label{c474}
	\end{align}
	By \eqref{c461}-\eqref{c464} and \eqref{c472}-\eqref{c474}, we have 
	\begin{align}
		p(x_{1234})=p(x_1,x_2)=p(x_1,x_3)=p(x_1,x_4).\label{c475}
	\end{align}
	Since $X_1$ is independent of $X_2,X_3$ and $X_4$ by \eqref{c465}-\eqref{c467}, together with\eqref{c475}
	\begin{align}
		p(x_1)p(x_2)=p(x_1)p(x_3)=p(x_1)p(x_4),\label{c476} 
	\end{align}
	which implies $p(x_2)=p(x_3)=p(x_4)$. Note that  $X_4$ is independent of $X_2$ and $ X_3$ by \eqref{c468} and \eqref{c469}, by Lemma \ref{lem1}, $X_i$ is uniformly distributed on $\mathcal{X}_{i}$  and $H(X_i)=\log v$ for $i=2,3,4$, where $v=|\mathcal{X}_{2}|=|\mathcal{X}_{3}|=|\mathcal{X}_{4}|$.
	As $\h\in F$ and $(X_i, i\in N_4)$ is
	its characterizing random vector, we have 
	\begin{align}
		H(X_2)=H(X_3)=H(X_4)=a+b. \label{c477}
	\end{align}
	Hence $a+b$ must take the value of  $\log v$.
	
	As for the ``if'' part of the theorem, we show that all the polymatroids $(a,b)$ satisfying $a+b=\log v$ are entropic. Let $p_1,p_2,\dots,p_v$ be positive value satisfying \rv{$H(p_1,p_2,\dots,p_v)=a$}.  Let $X_1$ be distributed on \rv{$\mathbb{I}_{v^2}$} such that 
	\begin{align}
		p_{X_1}(i)= \dfrac{ p_{\lceil \frac{i+1}{v}\rceil}}{v}, \text{ for } i=0,1,...,v^2-1. \label{c478}
	\end{align} 
	Let $X_2$ be independent of $X_1$ and uniformly distributed on $\mathbb{I}_{v}$.
	Let  
	\begin{align}
		X_3&=(\lfloor\frac{X_1}{v}\rfloor +X_2) \mod v, \label{c479}\\
		X_4&=(X_1\mod v +X_2) \mod v.\label{c480}
	\end{align}
	Then, it can be checked that the entropy function of such constructed $(X_i,i\in N_4)$ is $\h=(a,b)$.
	%
	%
\end{proof}

\subsection{Non-entropic faces}
\label{J}
\begin{theorem}
	\label{rk31}
	For $F=(V^{12}_8,E)$ with $E=U^{1}_{1,1},U^{3}_{1,1},U^{13}_{1,2},U^{134}_{1,3},U_{1,4},U^{123}_{2,3}$  or $U_{3,4}$, any $\h=(a,b)\in F$ is non-entropic if $a$ and $b$ are both positive.
\end{theorem}

\begin{proof}
 It can be immediately proven using \cite[Thm. 11]{dougherty2011nonshannon} that the extreme $1$-dimensional segment involving the top of the pyramid $(V^\alpha_8)$ is not  in $\overline{\Gamma_4^*}$.
\end{proof}

\bibliographystyle{IEEEtran0}
\bibliography{reference}

\end{document}